\documentclass[a4paper,USenglish]{lipics-v2021}
\hideLIPIcs 
\usepackage[T1]{fontenc}
\usepackage[utf8]{inputenc}
\usepackage[dvipsnames]{xcolor}
\usepackage{graphicx}
\usepackage{thm-restate}
\usepackage{amsmath}
\usepackage{amssymb}
\usepackage{bbm}
\usepackage{cases}
\usepackage{hyperref, graphicx}
\usepackage{enumitem}
\usepackage{amsthm}
\newtheorem*{theorem*}{Theorem}

\usepackage[dvipsnames]{xcolor}
\usepackage{graphicx}

\usepackage{tikz}
\usetikzlibrary{arrows.meta}
\tikzset{
  arr/.style={{}-Latex,shorten <=-2pt}
}
\usetikzlibrary{automata,calc,positioning}

\DeclareMathOperator*{\argmin}{argmin}
\newcommand{\FPA}{\Phi^{(t)}_{\mbox{\rm\tiny all}}}

\newcommand{\CCA}{\Omega^{(t)}_{\mbox{\rm\tiny all}}}
\newcommand{\ecc}{\mbox{\rm ecc}}
\newcommand{\radius}{\mbox{\rm radius}}
\newcommand{\sux}{\mbox{\rm succ}}
\newcommand{\view}{\mbox{\rm view}}
\newcommand{\cc}{\mbox{\rm comp}}

\newcommand{\config}{\mbox{\rm config}}
\newcommand{\dist}{\mbox{\rm dist}}
\newcommand{\env}{\mbox{\rm env}}

\newcommand{\IF}{\mathsf{IF}}
\newcommand{\KA}{{\rm \textsf{KNOW-ALL}}}

\usepackage{todonotes}

\title{Agreement Tasks in Fault-Prone Synchronous Networks of Arbitrary Structure}
\titlerunning{Agreement Tasks in Fault-Prone Synchronous Networks of Arbitrary Structure}

\author{Pierre Fraigniaud}{Institut de Recherche en Informatique Fondamentale (IRIF)\\ CNRS and Université Paris Cité, France \and \url{https://www.irif.fr/\~{}pierref} }{pierre.fraigniaud@irif.fr}{https://orcid.org/0000-0003-4534-4803}{Additional support from  ANR projects DUCAT (ANR-20-CE48-0006), ENEDISC, and QuDATA (ANR-18-CE47-0010). }

\author{Minh Hang Nguyen}{Institut de Recherche en Informatique Fondamentale (IRIF)\\ CNRS and Université Paris Cité, France \and \url{https://www.irif.fr/~mhnguyen/} }{mhnguyen@irif.fr}{https://orcid.org/0009-0008-2391-029X}{Additional support from  ANR projects DUCAT (ANR-20-CE48-0006), TEMPORAL (ANR-22-CE48-0001), and ENEDISC, and by the European Union’s Horizon 2020 program H2020‑MSCA ‑COFUND‑2019 Grant agreement n° 945332.}

\author{Ami Paz}{Laboratoire Interdisciplinaire des Sciences du Numérique (LISN)\\
CNRS and Université Paris-Saclay, France \and \url{sites.google.com/view/amipaz/}}{ami.paz@lisn.fr}{https://orcid.org/0000-0002-6629-8335}{}

\authorrunning{P. Fraigniaud, M. H. Nguyen, and A. Paz}

\Copyright{Pierre Fraigniaud, Minh Hang Nguyen, and Ami Paz} 



\acknowledgements{The authors thank Stephan Felber, Mikaël Rabie, Hugo Rincon Galeana, and Ulrich Schmid for fruitful discussions on this paper.}

\nolinenumbers

\begin{document}

\maketitle

\begin{abstract}
Consensus is arguably the most studied problem in distributed computing as a whole, and particularly in the distributed message-passing setting.
In this latter framework, research on consensus has considered various hypotheses regarding the failure types, the memory constraints, the algorithmic performances (e.g., early stopping and obliviousness), etc.
Surprisingly, almost all of this work assumes that messages are passed in a \emph{complete} network, i.e., each process has a direct link to every other process.
A noticeable exception is the recent work of Castañeda et al. (Inf. Comput. 2023) who designed a generic oblivious algorithm for consensus running in $\radius(G,t)$ rounds in every graph~$G$, when up to $t$ nodes can crash  by irrevocably stopping, where $t$ is smaller than the node-connectivity $\kappa$ of~$G$. Here, $\radius(G,t)$ denotes a graph parameter called the \emph{radius of~$G$ whenever up to $t$ nodes can crash}. For $t=0$, this parameter coincides with $\radius(G)$, the standard radius of a graph, and, for $G=K_n$, the running time $\radius(K_n,t)=t +1$ of the algorithm exactly matches the known round-complexity of consensus in the clique~$K_n$.

Our main result is a proof that $\radius(G,t)$ rounds are necessary for oblivious algorithms solving consensus in $G$ when up to $t$ nodes can crash, thus validating a conjecture of Castañeda et al., and demonstrating that their consensus algorithm is optimal for any graph~$G$. We also extend the result of Castañeda et al. to two different settings: First, to the case where the number $t$ of failures is not necessarily smaller than the connectivity $\kappa$ of the considered graph; Second, to the $k$-set agreement problem for which agreement is not restricted to be on a single value as in consensus, but on up to $k$ different values.  
\end{abstract}

\newpage
\tableofcontents
\newpage
\section{Introduction}

For $t\geq 0$, the standard \emph{synchronous $t$-resilient message-passing} model assumes $n\geq 2$ nodes labeled from 1 to $n$, and connected as a clique, i.e., as a complete graph~$K_n$. Computation proceeds as a sequence of synchronous rounds, during which every node can send a message to each other node, receive the message sent by each other node, and perform some local computation. Up to $t$ nodes may crash during the execution of an algorithm. When a node $v$ crashes at some round~$r\geq 1$, it stops functioning after round $r$ and never recovers. 
Moreover, some  (possibly all) of the messages sent by $v$ at round~$r$ may be lost, that is, when $v$ crashes, messages sent by $v$ at round~$r$ may reach some neighbors, while other neighbors of $v$ may not hear from $v$ at round~$r$. 
This model has been extensively studied in the literature (see, e.g., \cite{attiya2004distributed,HerlihyKR2013,Lynch96,Raynal2010}). 
In particular, it is known that consensus can be solved in $t+1$ rounds in the $t$-resilient model~\cite{dolev1983authenticated}, and this is optimal for every $t<n-1$ as far as the worst-case complexity is concerned~\cite{aguilera1999simple,dolev1983authenticated}. Similarly, $k$-set agreement, in which the cardinality of the set of output values decided by the (correct) nodes must not exceed~$k$, is known to be solvable in $\lfloor t/k\rfloor+1$ rounds~\cite{chaudhuri1991towards}, and this worst-case complexity is also optimal~\cite{chaudhuri1993tight}. 

It is only very recently that the synchronous $t$-resilient message-passing model has been extended to the setting in which the complete communication graph $K_n$ is replaced by an arbitrary communication graph~$G$ (see~\cite{CastanedaFPRRT23,ChlebusKOO23}). 
Specifically, the graph $G$ is fixed, but arbitrary, and the concern is to design algorithms for~$G$. 
It was proved in~\cite{CastanedaFPRRT23} that if the number of failures is smaller than the connectivity of the graph, i.e., if $t<\kappa(G)$, then consensus in $G$ can be solved in $\radius(G,t)$ rounds in the $t$-resilient model, where $\radius(G,t)$ generalizes the standard notion of graph radius to the scenarios in which up to $t$ nodes may fail by crashing. 
For $t=0$, $\radius(G,0)$ is the standard radius of the graph~$G$. For the complete graph $K_n$, the $\radius(K_n,t)$ upper bound from~\cite{CastanedaFPRRT23} coincides with the seminal $t+1$ upper bound for consensus in $K_n$.

To get an intuition of $\radius(G,t)$, let us consider the case of the $n$-node cycle~$C_n$, for $n\geq 3$. We have $\kappa(C_n)=2$, so we assume $t\leq 1$. The radius of $C_n$ is $\lfloor \frac{n}{2}\rfloor$, i.e., $\radius(C_n,0)=\lfloor \frac{n}{2}\rfloor$. For $t=1$, let $v$ be the node that crashes. We have $\radius(C_n,1)\geq n-2$, which is the distance between the two neighbors of~$v$ in $C_n$ if $v$ crashes ``cleanly'' at the first round, preventing them to communicate directly through~$v$. However, we actually have $\radius(C_n,1)=n-1$. Indeed, $v$ may crash at the first round, yet be capable to send a message to one of its neighbors, and this message needs $n-2$ additional rounds to reach the other neighbor of~$v$. That is, computing $\radius(G,t)$ requires to take into account not only which nodes crash, but when and how they are crashing --- by ``how'', it is meant that, for a node $v$ crashing at some round~$r$, to which neighbors they still succeed to communicate at this round, and to which they fail to communicate. 

Importantly, the algorithm of~\cite{CastanedaFPRRT23} is \emph{oblivious}, that is, the output of a node after $\radius(G,t)$ rounds is solely based on the set of pairs (\textit{node-identifier, input-value}) collected by that node during $\radius(G,t)$ rounds (and not, e.g., from whom, when, and how many times it received each of these pairs). There are many reasons why to restrict the study to oblivious algorithms. Among them, oblivious algorithms are simple by design, which is desirable for their potential implementation. Moreover, they are known to be efficient, as illustrated by the case of the complete graphs in which optimal solutions can be obtained thanks to oblivious algorithms. As far as this paper is concerned (and maybe also as far as \cite{CastanedaFPRRT23} is concerned) obliviousness is highly desirable for the design of \emph{generic} solutions, that is ``meta-algorithms'' that apply to each and every graph~$G$. In such algorithms, every node forwards pairs (\textit{node-identifier, input-value}) during a prescribed number of rounds (e.g., during $\radius(G,t)$ rounds in the generic algorithm from~\cite{CastanedaFPRRT23}), and then decides on an output value according to a simple function of the set of input values received during these rounds, without having to track of the sequence of rounds at which each pair was received, and from which neighbor(s). Last but not least, intermediate nodes do not need to send complex information about the history of each piece of information transmitted during the execution, hence reducing the bandwidth requirement of the algorithms. 

\subsection{Objective}

The question of the optimality of the consensus algorithm performing in $\radius(G,t)$ rounds in any fixed graph $G$ for every number $t\leq \kappa(G)$ of failures was however left open in~\cite{CastanedaFPRRT23}. 
It was conjectured there that, for every graph~$G$, and for every $0\leq t<\kappa(G)$, no oblivious algorithm can solve consensus in $G$ in less than $\radius(G,t)$ rounds,
but this was proved only for the specific case of \emph{symmetric} (a.k.a.~\emph{vertex-transitive}) graphs\footnote{A graph  $G=(V,E)$ is vertex-transitive if, for every two nodes $u\neq v$, there exists an automorphism $f$ of~$G$ (i.e., a permutation $f:V\to V$ preserving the edges and the non-edges of~$G$) such that $f(u)=v$.}. 
Although the class of symmetric graphs includes, e.g.,  the complete graphs~$K_n$, the cycles $C_n$, and the $d$-dimensional hypercubes~$Q_d$, a lower bound $\radius(G,t)$ for every graph~$G$ in this class does not come entirely as a surprise since all nodes of a symmetric graph have the same eccentricity (i.e., maximum distance to any other node, generalized to include crash failures). 
The fact that all nodes have the same eccentricity implies that they can merely be ordered according to their identifiers for selecting the output value from the received pairs (\textit{node-identifier, input-value}). 
Instead, if the graph is not symmetric, a node that received a pair (\textit{node-identifier, input-value}) after $\radius(G,t)$ rounds does not necessarily know whether all the nodes have received this pair, 
and thus the choice of the output value from the set of received pairs is more subtle. 
Not only the design of an upper bound is made harder, but it also makes the determination of a strong lower bound more involved. The main question addressed in this paper is therefore the following:
For every graph $G$, and every non-negative integer $t<\kappa(G)$, is there an oblivious algorithm solving consensus in $G$ in less than $\radius(G,t)$ rounds under the $t$-resilient model (i.e., when up to $t$ nodes may fail by crashing)? 

Moreover, the study in~\cite{CastanedaFPRRT23} let aside the design of a generic (oblivious) algorithm for solving the standard important relaxation of consensus, namely $k$-set agreement. (Recall that, in $k$-set agreement, the set of all  values  outputted by the nodes must be of cardinality at most~$k$.) In fact, several tools developed in~\cite{CastanedaFPRRT23} do not extend to $k$-set agreement. Our next step is therefore to question the ability to design a generic algorithm for solving $k$-set agreement in arbitrary graphs~$G$, for every $k>1$. 

Last but not least, the study in~\cite{CastanedaFPRRT23} assumed that the number $t$ of failures is smaller that the connectivity $\kappa(G)$ of the graph $G$ at hand. We question what can be said about the case where the number of failures may be larger, that is when $t\geq \kappa(G)$, for both consensus and $k$-set agreement? 

\subsection{Our Results}

We extend the investigation of the  $t$-resilient model in arbitrary graphs, in various complementary directions. 

\subparagraph{Lower Bounds for Consensus.}

We affirmatively prove the conjecture from~\cite{CastanedaFPRRT23} that their  consensus algorithm is indeed optimal (among oblivious algorithms) for \emph{every} graph~$G$, and not only for symmetric graphs. That is, we show that, for every graph~$G$, no oblivious algorithm can solve consensus in $G$ in less than $\radius(G,t)$ rounds under the $t$-resilient model. 
This result is achieved by revisiting the notion of \emph{information flow graph} defined in~\cite{CastanedaFPRRT23} for fixing some inaccuracies in the original definition. We present a more robust (and accurate) definition of information flow graph, and we provide a characterization of the number of rounds required to solve consensus as a function of some structural property of that graph. With this characterization at hand, we establish the optimality of the algorithm in~\cite{CastanedaFPRRT23} by showing that $\radius(G,t)$ rounds are necessary for the information flow graph to satisfy the desired property required for consensus solvability.  

\subparagraph{Set Agreement.}

We demonstrate the existence of a generic oblivious algorithm for $k$-set agreement, in any arbitrary (connected) graph~$G=(V,E)$. 
This algorithm is generic in the sense that it obeys a general structure: 
\begin{enumerate}
    \item flooding the graph with the input values of a predetermined  set of nodes $C\subseteq V$, for $R$ rounds, and 
    \item after $R$ rounds, letting every node~$v\in V$ pick the input value of the node~$u\in C$ with smallest identifier among all the nodes in $C$ received by~$v$.
\end{enumerate}
Of course, both $C$ and the number of rounds~$R$ depend on~$G$. 

We show that for every graph~$G$, every $t$ smaller than the node-connectivity $\kappa(G)$ of~$G$, and every $k\geq 1$, $k$-set agreement can be solved in $R=\radius(G,t,k)$ rounds, where $\radius(G,t,k)$ extends the standard notion of graph radius to the case in which there are $k$ centers, and whenever up to $t$ nodes can crash. 

For $t=0$ and $k=1$, $\radius(G,t,k)$ coincides with the standard radius of~$G$. Moreover, for $k=1$, $\radius(G,t,1)=\radius(G,t)$. 
More specifically, like in the $k$-center problem, let us consider broadcast in $G$ from a set $S\subseteq V$ of $k$ nodes by flooding. Then $\radius(G,t,k)$ essentially denotes the minimum, taken over all sets~$S$ of $k$ nodes, of the broadcast time of~$S$, i.e., of the smallest number of rounds sufficient to guarantee that every non-faulty node receives information from at least one node in~$S$. The definition is a bit more subtle though, as the broadcast time of~$S$ actually depend on the failure pattern (i.e., which nodes crash, and when), and it may even be the case that $S$ cannot broadcast at all for some failure patterns (e.g., whenever all nodes in $S$ crash at the first round without sending any messages to their neighbors). So $\radius(G,t,k)$ is in fact defined as the minimum, taken over all set $S\subseteq V$ of cardinality at most~$k$, of the \emph{finite} broadcast time of~$S$. That is, in the definition of  $\radius(G,t,k)$, we ignore the combinations of source sets and failure patterns where the source set cannot broadcast under the failure pattern. 
We shall show that this allows us to capture the round-complexity of our algorithm.

\subparagraph{Beyond the Connectivity Threshold.}

Finally, inspired by~\cite{ChlebusKOO23}, we extend the study of consensus and set agreement in the $t$-resilient model in arbitrary graphs to the case where the number $t$ of crash failures is arbitrary, i.e., not necessarily lower than the connectivity $\kappa(G)$ of the considered graph~$G$. 
We show that all our algorithms can be extended to this framework, at the mere cost of relaxing consensus and $k$-set agreement to impose agreement to hold within each connected component of the graph resulting from removing the faulty nodes from~$G$. 
Under this somehow unavoidable relaxation, we present extension of the consensus algorithm from~\cite{CastanedaFPRRT23} in particular, and of our $k$-set agreement algorithm in general, to $t$-resilient models for $t\geq \kappa(G)$, and express the round complexities of these algorithms in term of a non-trivial extension of the radius notion to disconnected graphs. 

\subsection{Related Work}

Distributed computing in synchronous networks has a long tradition, including the early studies of the message complexity and round complexity of various tasks such as leader election, spanning tree constructions, BFS and DFS traversals, etc.~(see, e.g., \cite{attiya2004distributed,Lynch96}). The topic has then  flourished in the 2000s under the umbrella of the so-called \textsf{LOCAL} and  \textsf{CONGEST} models~\cite{Suomela2020,Peleg2000}, with the study of numerous graph problems such as coloring, maximal independent set, minimum-weight spanning tree, etc. 

Distributed computing in synchronous \emph{fault-prone} networks has also a long history, but it remained for a long time mostly confined to the special case of the message-passing model in the complete networks. That is, $n$ nodes subject to \emph{crash} or \emph{malicious} (a.k.a.~Byzantine) failures are connected as a complete graph~$K_n$ in which every pair of nodes has a private reliable link allowing them to exchange messages.  In this setting, a significant amount  of effort has been dedicated to narrowing down the complexity of solving agreement tasks such as consensus and, more generally, $k$-set agreement for $k\geq 1$. This includes in particular the issue of \emph{early stopping} algorithms whose performances depend on the actual number of failures~$f$ experienced during the execution of the algorithm, and not on the upper bound~$t$ on the number of failures. We refer to a sequence of surveys on the matter~\cite{CastanedaMRR17,Raynal02,RaynalT06}.

In the Byzantine case,  general communication graphs were studied early on~\cite{Dolev82}, and are  still being investigated~\cite{KhanNV19}.
In the stop-fault case, on the other hand,
it is only recently that this approach has been extended to arbitrary networks, beyond the case of the complete graph~$K_n$~\cite{CastanedaFPRRT23,ChlebusKOO23}. 
Our paper is carrying on the preliminary investigations in~\cite{CastanedaFPRRT23}, by extending them from consensus to $k$-set agreement, establishing various lower bounds including one demonstrating the optimality of the consensus algorithm in~\cite{CastanedaFPRRT23}, and extending the analysis to the case where the number of crashes may exceed the connectivity threshold. 
The original work in~\cite{CastanedaFPRRT23} has been extended to solving consensus when \emph{links} are subject to crash failures~\cite{ChlebusKOO23}. Several consensus algorithms were proposed in~\cite{ChlebusKOO23}, but their round complexities are expressed as a function of the so-called \emph{stretch}, defined as the number of connected components of the graph after removing the faulty links, plus the sum of the diameters of the connected components. Instead, the round-complexity of the algorithm  in~\cite{CastanedaFPRRT23} is expressed in term of the \emph{radius}, which is a more refined measure. Indeed, we show that the upper bound in~\cite{CastanedaFPRRT23} is tight (no multiplicative constants, nor even additive constants). The consensus algorithms in~\cite{ChlebusKOO23} however extend to the case where failures may disconnect the graph, and the task is then referred to as ``disconnected agreement''. Again, the complexities of the algorithms are expressed in term of the stretch, while we shall express the complexity of our local consensus algorithm as a function of the more refined radius parameter. We actually conjecture that our local consensus algorithm is optimal (with no multiplicative nor additive constants) for all~$t$, no matter whether $t<\kappa(G)$ or $t\geq \kappa(G)$. On the other hand, some consensus algorithms proposed in~\cite{ChlebusKOO23} are early stopping, but the one with round-complexity close to the stretch of the actual failure pattern is not oblivious, and it uses messages with size significantly larger than the size of the messages in oblivious algorithms.   
 
The case of \emph{omission} failures has also attracted a lot of attention. 
In this context, nodes are reliable but messages may be lost. 
This is modeled as a sequence $\mathcal{S}=(G_i)_{i\geq 1}$ of directed graphs, where $G_i$ captures the connections that are functioning at round~$i$. 
The \emph{oblivious message adversary} model allows an adversary to choose each communication graph $G_i$ from a set $\mathcal{G}$ and independently of its choices for the other graphs.
The nodes know the set $\mathcal{G}$ a priori, but not the actual graph picked by the adversary at each round). We refer to \cite{coulouma2013characterization,nowak2019topological,winkler2024time} for recent advances in this domain, including solving consensus.
We also refer to the \emph{heard-of} model~\cite{Charron-BostM09,Charron-BostS09}, which bears similarities with the oblivious message adversary model. 

The case of \emph{transient} failures is addressed in the context of \emph{self-stabilizing} algorithms~\cite{Dolev2000}. As opposed to most distributed algorithms for networks, which start from a given specific initial configuration, self-stabilizing algorithms must be able to start from any initial configuration (which may result from a corruption of the internal variables of the nodes). Under the synchronous scheduler, a self-stabilizing algorithm performs in a sequence of synchronous rounds, just that it must be able to cope with an arbitrary initial state of the system.  

Last but not least, we underline the recent trend related to modeling communication between nodes (under the full-information paradigm) as a topological deformation of the input simplicial complex, and the computation (i.e., the decision of each node regarding its output value) as a simplicial map from the deformed input complex to the output simplicial complex~\cite{HerlihyKR2013}. The \KA\/ model~\cite{castaneda2021topological} has been designed as a first attempt to understand the \textsf{LOCAL} model through the lens of algebraic topology. In particular, it was shown that $k$-set agreement in a graph~$G$ known to all the nodes a priori requires $r$ rounds, where $r$ is the smallest integer such that there exists a $k$-node dominating set in the $r$-th transitive closure of~$G$.   
A follow-up work~\cite{FraigniaudP20}  minimized the involved simplicial complexes, and extended the framework to handle graph problems such as finding a proper coloring.

The study of \emph{anonymous} networks, in which nodes may not be provided with distinct identifiers, and of \emph{asynchronous} communication and computing, is beyond the scope of this paper, and we merely refer the reader to~\cite{Delporte-Gallet18,Delporte-GalletFT09,FraigniaudLR22,LunaV22,LunaV23} for recent advances in these domains, as far as computing in (non-necessarily complete) networks is concerned. 

\section{Model and definitions}

In this section, we recall the definition of the (synchronous) $t$-resilient model for networks, and the graph theoretical notions related to this model, all taken from~\cite{CastanedaFPRRT23}, as well as the consensus algorithm presented there. 

\subsection{The Model}

Let $G=(V,E)$ be an $n$-node undirected graph, which is also connected and simple (i.e., no multiple edges, nor self-loops). 
Each node $v\in V$ is a computing entity modeled as an infinite state machine. 
The nodes of $G$ have distinct identifiers, which are positive integers. For the sake of simplifying the notations, we shall not distinguish a node~$v$ from its identifier; 
for instance, by ``the smallest node'' we mean ``the node with the smallest identifier''. Initially, every node knows the graph~$G$, that is, it knows the identifiers of all nodes, and how the nodes are connected. 
The uncertainty is thus not related to the initial structure of the connections, but is only due to the presence of potential failures, in addition to the fact that, of course, every node is not a priori aware of the inputs of the other nodes. 

Computation in~$G$ proceeds as a sequence of synchronous rounds. 
All nodes start simultaneously, at round~1. 
At each round, each node sends a message to each of its neighbors in~$G$, receives the messages sent by its neighbors, and performs some local computation. 
Each node may however fail by crashing --- when a node crashes, it stops functioning and never recover.
However, if a node~$v$ crashes at round~$r$, it may still send a message to a non-empty subset of its set~$N(v)$ of neighbors during round~$r$.
For every positive integer $t\geq 0$, the $t$-resilient model assumes that at most $t$ nodes may crash. 
A \emph{failure pattern} is defined as a set 
\[
\varphi=\{(v,F_v,f_v)\mid v\in F\}
\]
where $F\subset V$ is the set of faulty nodes in~$\varphi$, with $0\leq |F|\leq t$, and, for each node $v\in F$, we use $f_v$ to specify the round at which $v$ crashes, and $F_v\subseteq N(v)$ to specify the non-empty set of neighbors to which $v$ fails to send messages at round~$f_v$. 

A node $v\in F$ such that $F_v=N(v)$ is said to crash \emph{cleanly} in~$\varphi$ (at round~$f_v$). 
All the nodes in $V\smallsetminus F$ are the correct nodes in~$\varphi$. The failure pattern in which no nodes fail is denoted by $\varphi_\varnothing$. 
The set of all failure patterns in which at most $t$ nodes fail is denoted by~$\FPA$. 
In any execution of an algorithm in graph $G$ under the $t$-resilient model, the nodes know~$t$ and~$G$, but they do not know in advance to which failure pattern they may be exposed. 
This absence of knowledge is the source of uncertainty in the $t$-resilient model.

\subsection{Eccentricity, connectivity, and radius}
\label{subsec:ecc-conn-radius}

The \emph{eccentricity} of a node $v$ in $G$ with respect to a failure pattern $\varphi$, denoted by $\ecc(v,\varphi)$, is defined as the minimum number of rounds required for broadcasting a message from~$v$ to all \emph{correct} nodes in $\varphi$. 
The broadcast protocol is by flooding, i.e., when a node receives a message at round~$r$, it forwards it to all its neighbors at round $r+1$. That is $\ecc(v,\varphi)$ is the maximum, taken over all correct nodes~$v'$, of the length of a shortest causal path from $v$ to $v'$, where a \emph{causal} path with respect to a failure pattern $\varphi$ from a node $v$ to a node~$v'$ is a sequence of nodes $u_1,\dots,u_q$ with $u_1=v$, $u_q=v'$, and, for every $i\in\{1,\dots,q-1\}$, $u_{i+1}\in N(u_i)$, $u_i$~has not crashed in $\varphi$ during rounds $1,\dots,i-1$, and if $u_i$ crashes in $\varphi$ at round~$i$, i.e., if $(u_i,F_i,i)\in \varphi$ for some non-empty set $F_i\subseteq N(u_i)$, then $u_{i+1}\notin F_i$. 

Note that $\ecc(v,\varphi)$ might be infinite, in case $v$ cannot broadcast to all correct nodes in~$G$ under~$\varphi$. 
A typical example is when $v$ crashes cleanly at the first round in~$\varphi$, before sending any message to any of its neighbors. 
A more elaborate failure pattern~$\varphi$  in which $v$ fails to broadcast is $\varphi=\{(v,N(v)\smallsetminus \{w\},1),(w,N(w),2)\}$ where $v$ crashes at round~1, and sends the message only to its neighbor~$w$, which crashes cleanly at round~2. 

The node-connectivity of~$G$, denoted~$\kappa(G)$, is the smallest integer $q$ such that removing $q$ nodes disconnects the graph~$G$ (or reduces it to a single node whenever $G$ is the complete graph~$K_n$). 
The following was established in~\cite{CastanedaFPRRT23}. 

\begin{proposition}[Lemma 1 in~\cite{CastanedaFPRRT23}]
    \label{lem:all-receive-iff-one-receive}
    For every graph $G$, every $t<\kappa(G)$, every node~$v$, and every failure pattern~$\varphi$ in the $t$-resilient model, $\ecc(v,\varphi)<\infty$ if and only if there exists at least one correct node that becomes aware of the message broadcast from~$v$. 
\end{proposition}

\noindent Note that, in particular, thanks to proposition~\ref{lem:all-receive-iff-one-receive}, if $v$ is correct then $\ecc(v,\varphi)<\infty$. Let 
\[
\Phi^\star_v=\{\varphi\in \FPA \mid \ecc(v,\varphi)<\infty\}
\]
denote the set of failure patterns in the $t$-resilient model in which $v$ eventually manages to broadcast to all correct nodes.  The \emph{$t$-resilient radius} is a key parameter defined in~\cite{CastanedaFPRRT23}: 

\begin{definition}
The \emph{$t$-resilient radius} of $G$ is 
\[
\radius(G,t)=\min_{v\in V}\max_{\varphi\in\Phi^\star_v} \ecc(v,\varphi). 
\]

\end{definition}

\subsection{Consensus, oblivious algorithms, and the information flow graph}

This section defines consensus, and survey the results in~\cite{CastanedaFPRRT23} regarding the round-complexity of oblivious consensus algorithms, which uses the notion of information flow graph. Note that this latter notion will be revisited, later in our paper. 

\subsubsection{Oblivious consensus algorithms}
\label{sec:algo-of-Castaneda}

In the consensus problem, every node $v\in V$ receives an input value $x_v$ from a set $I$ of cardinality at least~2, and every correct node must decide on an output value~$y_v\in I$ such that (1)~$y_u=y_v$ for every pair $\{u,v\}$ of correct nodes, and (2)~for every correct node $v\in V$, there exists $u\in V$ (not necessarily correct) such that $y_v=x_u$. 

Assuming that every node $u\in V$ starts broadcasting the pair $(u,x_u)$ at round~1, we let $\view(v,\varphi,r)$ be the \emph{view} of node~$v$ after $r\geq 0$ rounds in failure pattern~$\varphi$, that is, the set of pairs $(u,x_u)$ received by~$v$ after $r$ rounds. An algorithm solving consensus is said to be \emph{oblivious} if the output~$y_v$ of every correct node~$v$ depends only on the set of values received by~$v$ during the execution of the algorithm. That is, in an $r$-round oblivious algorithm executed under failure pattern~$\varphi$, every node $v$ outputs a value  based solely on the set of pairs $(u,x_u)\in \view(v,\varphi,r)$
(and not, say, on when each value was first received, or from which neighbor it was received). 
The following result was proved in~\cite{CastanedaFPRRT23}.

\begin{proposition}[Theorem~2 in \cite{CastanedaFPRRT23}]
\label{theo:mainCFRRT}
    For every graph $G$ and every $t<\kappa(G)$, consensus in $G$ can be solved by an oblivious algorithm running in $\radius(G,t)$ rounds under the $t$-resilient model. 
\end{proposition}
That is, consensus can be solved in the minimal time it takes for a \emph{fixed} node to broadcast in all failure patterns (in which it manages to broadcast). Note that 
$\radius(G,t)$ might be much larger than 
$\max_{\varphi\in\FPA}\min_{v\in V}\ecc(v,\varphi)$. 
For instance, the radius of the clique $K_n$ is $t+1$: consider a path $(v_1,...,v_{t+1})$ in which $v_1=v$, and, for every $i\in \{1,...,t\}$, $v_i$ crashes at round $i$ while sending only to $v_{i+1}$. 
On the other hand, 
$\max_{\varphi\in\FPA}\min_{v\in V}\ecc(v,\varphi)=1$  because, for every failure pattern~$\varphi$, there is a (correct) node~$v$ that broadcasts to all correct nodes in a single round. 
Similarly, the  cycle $C_n$ has radius $n-1$, whereas 
$\max_{\varphi\in\FPA}\min_{v\in V}\ecc(v,\varphi)$ 
is roughly~$n/2$.

The consensus algorithm in~\cite{CastanedaFPRRT23} works as follows. It selects an ordered set of $t+1$ nodes $s_1,\dots,s_{t+1}$ according to the following rules. Node~$s_1$ is a node with smallest eccentricity, i.e., a node that broadcasts the fastest among all nodes. However, there are failure patterns for which $s_1$ fails to broadcast (e.g., if $s_1$ crashes cleanly at round~1).  Node~$s_2$ is a node that broadcasts the fastest for all failure patterns in which $s_1$ fails to broadcast, that is node~$s_2$ is a node that broadcasts the fastest for all failure patterns in $\FPA\smallsetminus \Phi_{s_1}^\star$. Similarly, node $s_3$ is a node that broadcasts the fastest for all failure patterns in which $s_1$ and $s_2$ fail to broadcast, that is node~$s_3$ is a node that broadcasts the fastest for all failure patterns in $\FPA\smallsetminus (\Phi_{s_1}^\star\cup \Phi_{s_2}^\star)$. And so on, for every $1<i\leq t+1$, $s_i$ is a node that broadcasts the fastest for all failure patterns in 
\[
\FPA\smallsetminus \cup_{j=1,\dots,i-1}\Phi_{s_j}^\star.
\]
A key property of the sequence $s_1,\dots,s_{t+1}$ defined as above is that, for all $1<i\leq t+1$, the worst-case broadcast time of $s_i$ over all failure patterns in 
\[
\FPA\smallsetminus \cup_{j=1,\dots,i-1}\Phi_{s_j}^\star
\]
is at most the worst-case broadcast time of $s_{i-1}$ over all failure patterns in 
\[
\FPA\smallsetminus \cup_{j=1,\dots,i-2}\Phi_{s_j}^\star.
\]
As a consequence, for every $i\in\{1,\dots,t+1\}$, the worst-case broadcast time of $s_i$ over all failure patterns in $\FPA\smallsetminus \cup_{j=1,\dots,i-1}\Phi_{s_j}^\star$ is at most $\radius(G,t)$ rounds. 

The algorithm in~\cite{CastanedaFPRRT23} merely consists of letting all nodes $s_1,\dots,s_{t+1}$ broadcast the pairs $(s_i,x_{s_i})$ by flooding during $\radius(G,t)$ rounds. Every node $u$ then selects as output the input $x_{s_i}$ of the node $s_i$ with smallest index~$i$ such that the pair $(s_i,x_{s_i})$ was received by node~$u$. It was shown that this choice guarantees agreement.

\subsection{Information flow graph}

The lower bound from \cite{CastanedaFPRRT23} on the number of rounds for achieving consensus in vertex-transitive graphs used the core notion of  \emph{information flow digraph}. The (directed) graph $\IF(G,r)$ captures the state of mutual knowledge of the nodes at the end of round $r\geq 1$, assuming every node~$u$ broadcasts the pair $(u,x_u)$ by flooding throughout the graph~$G$, starting at round~1.

\begin{itemize}
    \item The vertices of $\IF(G,r)$ are all pairs $(v,\view(v,r,\varphi))$ for  $v\in V$ and $\varphi\in\FPA$ in which $v$ does not crash in $\varphi$ during the first $r$  rounds. Note that a same vertex of $\IF(G,r)$ can represent both $(v,\view(v,r,\varphi))$ and $(v,\view(v,r,\psi))$ if $v$ has the same view after $r$ rounds in $\varphi$ and $\psi$.
\item There is an arc from $(u,\view(u,r,\varphi))$ to $(v,\view(v,r,\varphi))$ whenever $(u,x_u)\in \view(v,r,\varphi)$, where~$x_u$ is the input of $u$.
\end{itemize}

The \emph{connected components} of $\IF(G,r)$ play an important role, where by connected component we actually refer to the vertices of a connected component of the undirected graph resulting from $\IF(G,r)$ by ignoring the directions of the arcs. A node $v\in V$ of the communication graph $G=(V,E)$ is said to \emph{dominate} a connected component $C$ of $\IF(G,r)$ if, for every vertex $(u,\view(u,r,\varphi))\in C$ with $u\neq v$ there is a vertex $(v,\view(v,r,\varphi))\in C$ with an arc from $(v,\view(v,r,\varphi))$ to $(u,\view(u,r,\varphi))$ in $\IF(G,r)$. 
The following result characterizes the round-complexity of consensus in~$G$.

\begin{proposition}[Theorem~3 in \cite{CastanedaFPRRT23}]
\label{theo:characterizationCFRRT}
    For every graph $G=(V,E)$ and every $t<\kappa(G)$, consensus in $G$ can be solved by an oblivious algorithm running in $r$ rounds under the $t$-resilient model if and only if every connected component of $\IF(G,r)$ has a dominating node in~$V$. 
\end{proposition}

It was proved in~\cite{CastanedaFPRRT23} that, if $G$ is a symmetric graph then no node in $V$ dominates $\IF(G,\radius(G,t)-1)$.
Property~\ref{theo:characterizationCFRRT} immediately implies  that consensus in $G$ cannot be solved by an oblivious algorithm running in less than $\radius(G,t)$ rounds under the $t$-resilient model. 
Their proof, however, holds only for symmetric graphs, and does not extend to general graphs. 

\subparagraph{Remark.}

The definition of the information flow \emph{digraph} in~\cite{CastanedaFPRRT23} actually suffers from inconsistencies, and Theorem~3 there is formally incorrect. 
Roughly, it overlooks the possibility of deciding on an input of a process that already stopped.
The ``spirit'' of the definition and the theorem is nevertheless plausible,
and the specific consequences mentioned there are correct.
For establishing our lower bound, we had to fix the inaccuracy in the definition of the information flow digraph, and the bugs in the proof of Theorem~3 of~\cite{CastanedaFPRRT23}.
In the next section we introduce a new information flow \emph{graph} instead of the digraph of~\cite{CastanedaFPRRT23}, and establish a correct version of Theorem~3 using that definition (cf.\ Theorem~\ref{thm:new-characterizationCFRRT}).

\section{Lower bounds for consensus}
\label{sec:lower-bound-for-consensus}

We show that the consensus algorithm in~\cite{CastanedaFPRRT23} is optimal for every graph~$G$, and not only for symmetric graphs. 
Specifically, we establish the following.

\begin{theorem}
\label{theo:main:consensus}
For every graph $G$ and every $t<\kappa(G)$, consensus in $G$ cannot be solved in less than $\radius(G,t)$ rounds by an oblivious algorithm in the $t$-resilient model. 
\end{theorem}

This result was conjectured in~\cite{CastanedaFPRRT23}, but only proved to be true for symmetric graphs. 
The class of symmetric graphs includes cliques, cycles and  hypercubes, but remains limited. Moreover, in symmetric graphs, for every two nodes $u$ and~$v$, \[
\ecc(u,\FPA)=\ecc(v,\FPA)=\radius(G,t),
\]
which implies that a naive algorithm for consensus in which every node outputs the input received from the node with smallest identifier performs in $\radius(G,t)$ rounds. The fact that $\radius(G,t)$ is a tight upper bound for consensus is thus not surprising for the family of symmetric graphs because, essentially, the choice of the $t+1$ nodes $s_1,\dots,s_{t+1}$ defined in Section~\ref{sec:algo-of-Castaneda} does not matter. 

Instead, for an arbitrary graph~$G$, two different nodes may have different eccentricities, which may differ by a multiplicative factor~2 at least. 
As a consequence, the choice of the source nodes $s_1,\dots,s_{t+1}$ whose input can be adopted as output by the other nodes matters, as well as the ordering of these nodes (in case a node receives the input of two different source nodes). 

\begin{figure}[tb]
    \centering
    \begin{tikzpicture}[scale = 1.07]

\tikzstyle{whitenode}=[circle,minimum size=0pt,inner sep=0pt,font=\scriptsize]
\tikzstyle{thicknode}=[circle,minimum size=0pt,inner sep=0pt]

\draw (0,0) node[thicknode] (a0)   [] {000..0};
\draw (2,0) node[thicknode] (a1)   [] {100..0};
\draw (4,0) node[thicknode] (a2)   [] {110..0};
\draw (5,0) node[whitenode] ()   [] {$\ldots$};
\draw (6,0) node[thicknode] (a3)   [] {1..100};
\draw (8,0) node[thicknode] (a4)   [] {1..110};
\draw (10,0) node[thicknode] (a5)   [] {1..111};

\draw (0,-0.5) node[whitenode] ()   [] {$I_0$};
\draw (2,-0.5) node[whitenode] ()   [] {$I_1$};
\draw (4,-0.5) node[whitenode] ()   [] {$I_2$};
\draw (6,-0.5) node[whitenode] ()   [] {$I_{n-2}$};
\draw (8,-0.5) node[whitenode] ()   [] {$I_{n-1}$};
\draw (10,-0.5) node[whitenode] ()   [] {$I_n$};

\path[dotted,draw] (a0) edge node[above,font=\scriptsize] {$w_1$} (a1);
\path[dotted,draw] (a1) edge node[above,,font=\scriptsize] {$w_2$} (a2);
\path[dotted,draw] (a3) edge node[above,,font=\scriptsize] {$w_{n-1}$} (a4);
\path[dotted,draw] (a4) edge node[above,,font=\scriptsize] {$w_n$} (a5);

\end{tikzpicture}
    \caption{Input configurations $I_0,\ldots,I_n$ of a graph $G=(V,E)$, where $V=\{v_1,\ldots,v_n\}$. } 
    \label{fig:consensus-one-dim}
\end{figure}

\subsection{A naive lower bound}

A naive lower bound for the round-complexity of consensus is the maximum, over all failure patterns, of the time it takes \emph{some} node to broadcast in the given pattern, obtained by switching the min and max operator in the definition of $\radius(G,t)$, i.e., 
\begin{equation}\label{eq:trivial-lower-bound-consensus}
\max_{\varphi\in\FPA} \min_{v\in V} \ecc(v,\varphi). 
\end{equation}
Indeed, for every failure pattern~$\varphi$, even binary consensus under failure pattern~$\varphi$ cannot be solved in less than $R(\varphi)=\min_{v\in V} \ecc(v,\varphi)$ rounds. The proof of this claim is by a standard indistinguishability argument. Specifically, let us assume, for the purpose of contradiction, that there is an  algorithm ALG solving consensus in $G=(V,E)$ under failure pattern~$\varphi$ in $R(\varphi)-1$ rounds. Let us order the nodes of $G$ as $v_1,\dots,v_n$ arbitrarily. Let us consider the input configuration~$I_0$ in which all nodes have input~0. For every $i=1,\dots,n$, we gradually change the input configuration as follows (see Figure~\ref{fig:consensus-one-dim}). Since $\ecc(v_i,\varphi)>R(\varphi)$, there exists a node~$w_i$ that does not receive the input of $v_i$ in ALG. Let us then switch the input of $v_i$ from~0 to~1, and denote by $I_i$ the resulting input configuration. Note that $I_n$ is the input configuration in which all nodes have input~1. Note also that, for every $i\in\{1,\dots,n\}$, node $w_i$ does not distinguish $I_{i-1}$ from $I_i$, and therefore ALG must output the same at $w_i$ in both input configurations. Since, for every $i\in\{1,\dots,n\}$, all nodes must output the same value for input configuration~$I_i$, we get that the consensus value returned by ALG for $I_0$ is the same as for~$I_n$, which contradicts the validity condition.

It was conjectured in~\cite{CastanedaFPRRT23} that, in the $t$-resilient model, consensus needs longer time than $\max_{\varphi\in\FPA} \min_{v\in V} \ecc(v,\varphi)$, and cannot be solved by an oblivious algorithm in less than $\radius(G,t)$ rounds, i.e., the time it takes a fixed node to broadcast. As said before, this conjecture was however proved only for vertex-transitive graphs.

\subsection{Information flow graph revisited}

In order to prove Theorem~\ref{theo:main:consensus}, we first establish a consistent notion of information flow graph, which can then be used to characterize consensus solvability, and we fix the bugs in the proof of Theorem~3 in~\cite{CastanedaFPRRT23} (see Proposition~\ref{theo:characterizationCFRRT}) resulting from inconsistencies in the original definition of the information flow digraph. 

The main issue with the notion of information flow \emph{digraph} $\IF(G,r)$ as defined  in~\cite{CastanedaFPRRT23} comes from the fact that this directed graph includes only vertices $(v,\view(v,r,\varphi))$ where $v$ has not crashed in $\varphi$ during rounds $1,\dots,r$. 
%
%
The main issue is related to the concept of domination, as defined in~\cite{CastanedaFPRRT23}. A vertex $v$ dominates a connected component $C$ of $\IF(G,r)$ if the set $\{(v, \view(v,r,\varphi)) \mid  \varphi \in \FPA\}$ dominates~$C$. 
This is too restrictive, as the correct nodes may agree on the input value of a node~$v$ that has already crashed. 
It follows that, for some failure pattern $\varphi$, the vertex $(v, \view(v,r,\varphi))$ may not be present in $\IF(G,r)$ (and therefore cannot dominate any other vertices of $\IF(G,r)$), whereas the nodes that are correct in~$\varphi$ may agree on the input value of~$v$. The characterization of Theorem~3 in~\cite{CastanedaFPRRT23} is therefore incorrect, even if the ``spirit'' of the characterization remains conceptually valid, as we shall show in this section. 

To provide an illustration of the problems resulting from the original definition of information flow digraph in~\cite{CastanedaFPRRT23}, let us clarify that this definition was aiming for capturing any subset $\Phi\subseteq \FPA$ of failure patterns (for instance the subset $\Phi$ of failure patterns in which nodes crash cleanly), in which case only the failure patterns $\varphi\in\Phi$ are considered. Let us then consider the scenario displayed on Fig.~\ref{fig:IFG}. The graph $G$ is a 6-node path plus a universal node~$v$. The set $\Phi=\{\varphi\}$ contains a single failure pattern $\varphi$ in which $v$ crashes cleanly at the second round.

Fig.~\ref{fig:IFG} displays $\IF(G,1,\{\varphi\})$ and $\IF(G,2,\{\varphi\})$ as defined in~\cite{CastanedaFPRRT23} (the direction of the arcs are omitted, each edge corresponding to two symmetric arcs). 
A vertex $(v,\view(v,r,\varphi))$ is present in the former but not in the latter, and thus, as opposed to what one might expect since nodes acquire more and more information as time passes, $\IF(G,2,\{\varphi\})$ is not a denser super graph of $\IF(G,1,\{\varphi\})$ nor it includes more vertices (with larger views), as some vertices present in $\IF(G,1,\{\varphi\})$ may disappear in $\IF(G,2,\{\varphi\})$. 
In fact, node $v$ dominates $\IF(G,1,\{\varphi\})$, but it does not dominate $\IF(G,2,\{\varphi\})$. 
Therefore, when analyzing $G$ with the set 
$\{\varphi\}$ of failure patterns
using the characterization theorem in~\cite{CastanedaFPRRT23},
consensus should be solvable in 1~round but not in 2~rounds! 

\begin{figure}[tb]
	\centering
	\includegraphics[width=14cm]{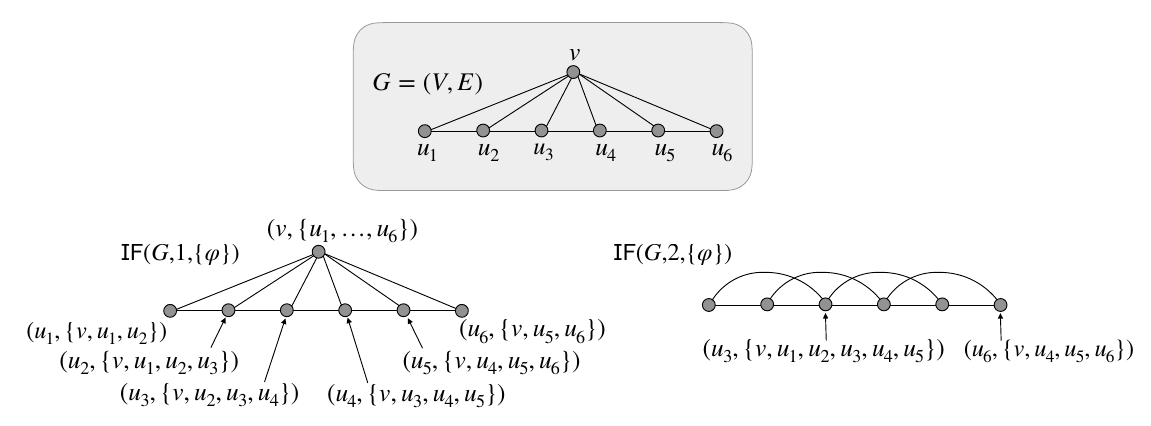}
	\caption{The information flow graph $\IF(G,r,\{\varphi\})$ as defined in~\cite{CastanedaFPRRT23} for $r=1$ and $r=2$, where $\varphi$ is the failure pattern in which $v$ crashes cleanly at the second round.
		No node dominates $\IF(G,2,\{\varphi\})$ (right), even though consensus is solvable in $G$ under $\varphi$ in $2$ rounds.}
	\label{fig:IFG}
\end{figure}

We propose below a more robust notion of information flow \emph{graph} (which is not directed anymore). The reader familiar with the algebraic topology interpretation of distributed computing~\cite{HerlihyKR2013} will recognize the mere 1-skeleton of the protocol complex after $r$ rounds. For the purpose of fixing the issues in~\cite{CastanedaFPRRT23}, we introduce $\IF(G,r,\Phi)$ for an arbitrary set of failure patterns $\Phi\subseteq\FPA$. 

\begin{definition}
	The information flow graph of a communication graph $G=(V,E)$ after $r\geq 0$ rounds for a set $\Phi\subseteq \FPA$, $t\geq 0$, of failure patterns is the graph $\IF(G,r,\Phi)$ defined as follows. 
	\begin{itemize}
		\item The vertices of $\IF(G,r,\Phi)$ are all pairs $(v,\view(v,r,\varphi))$ for  $v\in V$ and $\varphi\in\Phi$, where $v$ is correct in $\varphi$. 
		\item There is an edge between $(v_1,w_1)$ and $(v_2,w_2)$ in $\IF(G,r,\Phi)$ whenever there exists $\varphi\in\Phi$ such that $w_1 = \view(v_1,r,\varphi)$ and $w_2 = \view(v_2,r,\varphi)$
	\end{itemize}
\end{definition}

\subparagraph{Remark.} 

Unlike the definition of~\cite{CastanedaFPRRT23}, this new notion of information-flow graph is not limit limited to $t\leq \kappa(G)$. 

\medbreak

Note that a same vertex $(v,\omega)$ of $\IF(G,\Phi,r)$ can represent both $(v,\view(v,r,\varphi))$ and $(v,\view(v,r,\psi))$ if $v$ has the same view after $r$ rounds in $\varphi\in\Phi$ and $\psi\in\Phi$. Note also that, for every $\varphi\in\Phi$, the set  
\[
\config(G,r,\varphi)=\{(v,\view(v,r,\varphi))\in \IF(G,r,\Phi) \mid v\in V\}
\]
is a clique in $\IF(G,r,\Phi)$. The \emph{connected components} of $\IF(G,r,\Phi)$ play an important role, w.r.t. the following concept of \emph{domination}. 

\begin{definition}
	A node $v\in V$ of the communication graph $G=(V,E)$ is said to \emph{dominate} a connected component $C$ of $\IF(G,r,\Phi)$ if, for every $\varphi\in\Phi$ and every $u\in V$, 
	$$
	(u,\view(u,r,\varphi))\in C \Longrightarrow (v,x_v) \in \view(u,r,\varphi).
	$$
\end{definition} 

\noindent 
Note that only correct nodes need to be dominated, as 
\[
(u,\view(u,r,\varphi))\in C \subseteq \IF(G,r,\Phi)
\]
implies that $u$ is correct at round $r$.
On the other hand, any node may be dominating.
The following result characterizes the round-complexity of consensus in~$G$ by fixing the aforementioned inaccuracies in the definition of the information flow graph in~\cite{CastanedaFPRRT23}, with impact on the proof of their characterization theorem (Theorem~3 in~\cite{CastanedaFPRRT23}).  

\begin{theorem}
	\label{thm:new-characterizationCFRRT}
	For every graph $G=(V,E)$, every $t\geq 0$, and every set of failure patterns $\Phi\subseteq\FPA$, consensus in $G$ can be solved by an oblivious algorithm running in $r$ rounds under the $t$-resilient model with failure patterns in~$\Phi$ if and only if every connected component of $\IF(G,r,\Phi)$ has a dominating node in~$V$. 
\end{theorem}

\begin{proof}
	Let us first show that if every connected component of $\IF(G,r,\Phi)$ has a dominating node in~$V$ then consensus in $G$ can be solved by an oblivious algorithm running in $r$ rounds. For every connected component $C$ of $\IF(G,r,\Phi)$, let $v_C\in V$ be a node of $G$ that dominates~$C$. The algorithm proceeds as follows. Every node $v_C$ broadcasts by flooding during $r$ rounds. After $r$ rounds, every correct node $u$ considers its view, denoted by $\view(u)$. A crucial point is that $\view(u)$ may not be sufficient for $u$ to determine what is the actual failure pattern $\varphi\in\Phi$ experienced during the execution, merely because one may have 
	\[
	\view(u)=\view(u,r,\varphi)=\view(u,r,\psi)
	\]
	for two different failure patterns $\varphi,\psi$ in~$\Phi$. However, $\view(u)$ is sufficient to determine the connected component $C$ of $\IF(G,r,\Phi)$ to which $(u,\view(u))$ belongs.  Node $u$ outputs the input $x_{v_C}$ of node $v_C$. 
	
	To establish correctness of this algorithm, observe first that $(v_C,x_{v_C})$ belongs to the view of node~$u$. To see why, let $\varphi\in \Phi$, and let us consider the execution of the algorithm under~$\varphi$. Let $C$ be the connected component of $(u,\view(u,r,\varphi))$. Since $v_C$ dominates~$C$, the mere definition of domination implies that $(v_C,x_{v_C}) \in \view(u,r,\varphi)$. As a consequence, the algorithm is well defined. To show agreement, let $u'\neq u$ be another correct node in~$\varphi$. By definition of the information flow graph, there is an edge between $(u,\view(u,r,\varphi))$ and $(u',\view(u',r,\varphi))$, and thus these two vertices belong to the same connected component~$C$, and both output the same value~$x_{v_C}$.  
	
	\medbreak 
	
	For the other direction, we show the contrapositive. That is, we let $C$ be a connected component of $\IF(G,r,\Phi)$ that is not dominated, and we aim at showing that there are no oblivious  consensus algorithms in $G$ running in~$r$ rounds. Let us assume, for the purpose of contradiction, that there exists an oblivious  consensus algorithm ALG in $G$ running in~$r$ rounds. 
	
	\begin{claim}\label{claim:they-output-the-same}
		Let $(u, \view(u, r, \varphi))$ and $(u', \view(u', r, \varphi'))$ be two vertices of~$C$, where $u$ and $u'$ need not be different, nor do $\varphi$ and $\varphi'$.
		For the same input configuration, node $u$  outputs the same value in ALG under $\varphi$ as node $u'$ under $\varphi'$.
	\end{claim}
	
	To see why this claim holds, observe that, since $(u,\view(u,r,\varphi))$ and  $(u',\view(u',r,\varphi'))$ belong to the same connected component~$C$, there is a sequence 
	\[
	(v_0,\view(v_0,r,\psi_0)), \dots, (v_k,\view(v_k,r,\psi_k))
	\]
	of vertices of $C$ such that 
	\[
	(v_0,\view(v_0,r,\psi_0))=(u,\view(u,r,\varphi)), \;\; (v_k,\view(v_k,r,\psi_k))=(u',\view(u',r,\varphi')),
	\]
	and, for every $i\in\{0,\dots,k-1\}$, there is an edge between the two vertices $(v_i,\view(v_i,r,\psi_i))$ and $(v_{i+1},\view(v_{i+1},r,\psi_{i+1}))$ in $\IF(G,r,\Phi)$. 
	Note that, for every $i \in \{0,\ldots,k\}$, node $v_i$ is correct in $\psi_i$ since $(v_i,\view(v_i,r,\psi_i))$ belongs to the information flow graph.
	For every $i\in\{0,\dots,k-1\}$, the presence of an edge between $(v_i,\view(v_i,r,\psi_i))$ and $(v_{i+1},\view(v_{i+1},r,\psi_{i+1}))$ implies that there exists $\chi\in\Phi$ such that 
	\[
	(v_i,\view(v_i,r,\psi_i))=(v_i,\view(v_i,r,\chi)),
	\]
	and
	\[
	(v_{i+1},\view(v_{i+1},r,\psi_{i+1}))=(v_{i+1},\view(v_{i+1},r,\chi)). 
	\]
	As a consequence, since ALG is a consensus algorithm, ALG outputs the same value at $v_{i+1}$ under $\psi_{i+1}$ as it outputs at $v_i$ under~$\psi_i$, which is the value outputted by ALG under~$\chi$. Since this holds for every $i\in\{0,\dots,k-1\}$, we get that, in particular, $u$ outputs the same value in $\varphi$ as $u'$ in $\varphi'$, as claimed.  
	
	\medbreak
	
	For establishing a contradiction, let us enumerate the $n$ nodes of~$G$ as $u_0,\dots,u_{n-1}$  in arbitrary order. Since $C$ is not dominated, for every node~$u_i$, $i\in \{0,\dots,n-1\}$, there exists a vertex $(v_i,\view(v_i,r,\varphi_i))$ of $C$ such that $(u_i,x_{u_i})\notin \view(v_i,r,\varphi_i)$, where $v_i$ is correct in $\varphi_i$.
	For $i\in\{0,\dots,n\}$, let us denote by $I_i$ the input configuration in which the $n-i$ nodes $u_0,\dots,u_{n-(i+1)}$ have input~0, and all the other nodes have input~1. Thus, in particular, $I_0$ is the configuration in which all nodes have input~0, and $I_n$ is the configuration in which all nodes have input~1. Since, for every $i\in \{0,\dots,n-1\}$, $(u_i,x_{u_i})\notin \view(v_i,r,\varphi_i)$, node $u_i$ does not distinguish $I_i$ from $I_{i+1}$ under~$\varphi_i$, and thus ALG must output the same at $u_i$ for both configurations. 
	
	Since consensus imposes that all (correct) nodes output the same value, this means that, for every $i\in \{0,\dots,n-1\}$, all nodes output the same in ALG for $I_i$ and $I_{i+1}$ under~$\varphi_i$. By Claim~\ref{claim:they-output-the-same}, all nodes output the same for $I_i$ under $\varphi_i$ as they do for $I_{i+1}$  under~$\varphi_{i+1}$. It follows that all nodes output the same for $I_0$ under $\varphi_0$ as for $I_n$ under~$\varphi_n$. This is a contradiction as all nodes must output~0 for~$I_0$, whereas all nodes must output~1 for~$I_n$. 
\end{proof}

\subparagraph{Notation.}

For a fixed upper bound $t$ on the number of failures, for every graph $G$, and for every integer $r\geq 0$, we denote by $\IF(G,r)$ the information flow graph for the set of all failure patterns in the $t$-resilient model, that is, 
\[
\IF(G,r)=\IF(G,r,\FPA).
\]

\subsection{Proof of our lower bound}

To prove Theorem~\ref{theo:main:consensus}, we define the notion of \emph{successor} of a failure pattern. Given $\varphi\in\FPA$, 
we say that a node $u$ is \emph{crashing last} in $\varphi$ if there exists a triple $(u,F_u,f_u)\in \varphi$ (i.e., $u$ crashes in~$\varphi$), and, for every $(v,F_v,f_v)\in \varphi$, $f_u \geq f_v$. 

\begin{definition} \label{def:successor}
	Let $\varphi\in\FPA$,  let $(u,F_u,f_u)\in \varphi$, and assume that $u$ is crashing last in $\varphi$. A \emph{successor of $\varphi$ with respect to~$u$} is a failure pattern
	$$
	\sux(\varphi,u)= \Big ( \varphi  \smallsetminus \{(u,F_u,f_u)\} \Big ) \cup \{(u,F'_u,f'_u)\}
	$$
	where $F'_u$ and $f'_u$ are defined as follows (see Fig.~\ref{fig:successor}): 
	\begin{enumerate}
		\item If $F_u$ contains only faulty nodes in $\varphi$, then $f'_u = f_u +1$, and $F'_u = N(u) \smallsetminus \{ w \}$ for some arbitrary correct neighbor $w$ of~$u$. 
		
		\item If $F_u$ contains exactly one correct node $w$ in $\varphi$, then $f'_u = f_u + 1$, and $F'_u = N(u)$.
		
		\item If $F_u$ contains at least two correct nodes in $\varphi$, then $f'_u =f_u$, and $F'_u = F_u \smallsetminus \{ w \}$ for some arbitrary correct node $w \in F_u$.
	\end{enumerate}
\end{definition}


\begin{figure}[tb]
	\begin{subfigure}[b]{.07\textwidth}
\centering
\begin{tikzpicture}[scale=0.7]
    \tikzstyle{nonenode}=[circle,fill=white,minimum size=0pt,inner sep=0pt]
    \draw (2.5,11) node[nonenode] (a)   [] {};
    \draw (3,7) node[nonenode] (c)   [] {\large $\varphi$};
    \draw (3,0.5) node[nonenode] (c)   [] {\large $\varphi'$};
    \draw (3.5,-2.5) node[nonenode] (z)   [] {};
\end{tikzpicture}
\end{subfigure}
\begin{subfigure}[b]{.3\textwidth}
\centering
\begin{tikzpicture}[scale=0.7]
    \tikzstyle{nonenode}=[circle,fill=white,minimum size=0pt,inner sep=0pt]
    \tikzstyle{rednode}=[circle,draw,fill=red,minimum size=0pt,inner sep=2pt]
    \tikzstyle{whitenode}=[circle,draw,fill=white,minimum size=0pt,inner sep=2pt]
    \draw (1.5,11) node[nonenode] (c)   [] {Case 1: $f'_u=f_u+1$};
    \draw (0,7) node[rednode] (u)   [] {};
    \draw (3,9.5) node[rednode] (a)   [] {};
    \draw (3,8.5) node[rednode] (b)   [] {};
    \draw (3,7.5) node[rednode] (c)   [] {};
    \draw (3,6.5) node[rednode] (x)   [] {};
    \draw (3,5.5) node[whitenode] (y)   [] {};
    \draw (3,4.5) node[whitenode] (z)   [] {};
    \draw (3.5,4.5) node[nonenode] ()   [] {w};
    \draw (0,7.4) node[circle,fill=white,minimum size=0pt,inner sep=0pt] ()   [] {u};
    \foreach \x in {a,b,c,x,y,z}
      \draw (u)--(\x);
    \draw[draw=black] (2.5,7) rectangle ++(1,3);
    \draw (3.5,10) node[nonenode] ()   [] {\large $F_u$};
    \draw (-0.5,10.5) node[nonenode] ()   [] {};
    \draw (4,-2) node[nonenode] ()   [] {};

    \draw (0,0.5) node[rednode] (u')   [] {};
    \draw (3,3) node[rednode] (a')   [] {};
    \draw (3,2) node[rednode] (b')   [] {};
    \draw (3,1) node[rednode] (c')   [] {};
    \draw (3,0) node[rednode] (x')   [] {};
    \draw (3,-1) node[whitenode] (y')   [] {};
    \draw (3,-2) node[whitenode] (z')   [] {};
    \draw (3.5,-2) node[nonenode] ()   [] {w};
    \draw (0,0.9) node[circle,fill=white,minimum size=0pt,inner sep=0pt] ()   [] {u};
    \foreach \x in {a',b',c',x',y',z'}
      \draw (u')--(\x);
    \draw[draw=black] (2.5,-1.4) rectangle ++(1,5);
    \draw (3.5,3.5) node[nonenode] ()   [] {\large $F'_u$};
    \draw (-0.5,4) node[nonenode] ()   [] {};
    \draw (4,-2.5) node[nonenode] ()   [] {};
\end{tikzpicture}
\end{subfigure}
\begin{subfigure}[b]{.3\textwidth}
\centering
\begin{tikzpicture}[scale=0.7]
    \tikzstyle{nonenode}=[circle,fill=white,minimum size=0pt,inner sep=0pt]
    \tikzstyle{rednode}=[circle,draw,fill=red,minimum size=0pt,inner sep=2pt]
    \tikzstyle{whitenode}=[circle,draw,fill=white,minimum size=0pt,inner sep=2pt]
    \draw (1.5,11) node[nonenode] (c)   [] {Case 2: $f'_u=f_u+1$};
    \draw (0,7) node[rednode] (u)   [] {};
    \draw (3,9.5) node[rednode] (a)   [] {};
    \draw (3,8.5) node[rednode] (b)   [] {};
    \draw (3,7.5) node[whitenode] (c)   [] {};
    \draw (3,6.5) node[rednode] (x)   [] {};
    \draw (3,5.5) node[whitenode] (y)   [] {};
    \draw (3,4.5) node[whitenode] (z)   [] {};
    \draw (3.3,7.8) node[nonenode] ()   [] {w};
    \draw (0,7.4) node[circle,fill=white,minimum size=0pt,inner sep=0pt] ()   [] {u};
    \foreach \x in {a,b,c,x,y,z}
      \draw (u)--(\x);
    \draw[draw=black] (2.5,7) rectangle ++(1,3);
    \draw (3.5,10) node[nonenode] ()   [] {\large $F_u$};
    \draw (-0.5,10.5) node[nonenode] ()   [] {};
    \draw (4,-2) node[nonenode] ()   [] {};

    \draw (0,0.5) node[rednode] (u')   [] {};
    \draw (3,3) node[rednode] (a')   [] {};
    \draw (3,2) node[rednode] (b')   [] {};
    \draw (3,1) node[whitenode] (c')   [] {};
    \draw (3,0) node[rednode] (x')   [] {};
    \draw (3,-1) node[whitenode] (y')   [] {};
    \draw (3,-2) node[whitenode] (z')   [] {};
    \draw (3.3,1.3) node[circle,fill=white,minimum size=0pt,inner sep=0pt] ()   [] {w};
    \draw (0,0.9) node[circle,fill=white,minimum size=0pt,inner sep=0pt] ()   [] {u};
    \foreach \x in {a',b',c',x',y',z'}
      \draw (u')--(\x);
    \draw[draw=black] (2.5,-2.5) rectangle ++(1,6);
    \draw (3.5,3.5) node[nonenode] ()   [] {\large $F'_u$};
    \draw (-0.5,4) node[nonenode] ()   [] {};
    \draw (4,-2.5) node[nonenode] ()   [] {};
\end{tikzpicture}
\end{subfigure}
\begin{subfigure}[b]{.3\textwidth}
\centering
\begin{tikzpicture}[scale=0.7]
    \tikzstyle{nonenode}=[circle,fill=white,minimum size=0pt,inner sep=0pt]
    \tikzstyle{rednode}=[circle,draw,fill=red,minimum size=0pt,inner sep=2pt]
    \tikzstyle{whitenode}=[circle,draw,fill=white,minimum size=0pt,inner sep=2pt]
    \draw (1.5,11) node[nonenode] (c)   [] {Case 3: $f'_u=f_u$};
    \draw (0,7) node[rednode] (u)   [] {};
    \draw (3,9.5) node[rednode] (a)   [] {};
    \draw (3,8.5) node[whitenode] (b)   [] {};
    \draw (3,7.5) node[whitenode] (c)   [] {};
    \draw (3,6.5) node[rednode] (x)   [] {};
    \draw (3,5.5) node[whitenode] (y)   [] {};
    \draw (3,4.5) node[whitenode] (z)   [] {};
    \draw (3.3,7.8) node[nonenode] ()   [] {w};
    \draw (0,7.4) node[circle,fill=white,minimum size=0pt,inner sep=0pt] ()   [] {u};
    \foreach \x in {a,b,c,x,y,z}
      \draw (u)--(\x);
    \draw[draw=black] (2.5,7) rectangle ++(1,3);
    \draw (3.5,10) node[nonenode] ()   [] {\large $F_u$};
    \draw (-0.5,10.5) node[nonenode] ()   [] {};
    \draw (4,-2) node[nonenode] ()   [] {};

    \draw (0,0.5) node[rednode] (u')   [] {};
    \draw (3,3) node[rednode] (a')   [] {};
    \draw (3,2) node[whitenode] (b')   [] {};
    \draw (3,1) node[whitenode] (c')   [] {};
    \draw (3,0) node[rednode] (x')   [] {};
    \draw (3,-1) node[whitenode] (y')   [] {};
    \draw (3,-2) node[whitenode] (z')   [] {};
    \draw (3.3,1.3) node[circle,fill=white,minimum size=0pt,inner sep=0pt] ()   [] {w};
    \draw (0,0.9) node[circle,fill=white,minimum size=0pt,inner sep=0pt] ()   [] {u};
    \foreach \x in {a',b',c',x',y',z'}
      \draw (u')--(\x);
    \draw[draw=black] (2.5,1.6) rectangle ++(1,1.9);
    \draw (3.5,3.5) node[nonenode] ()   [] {\large $F'_u$};
    \draw (-0.5,4) node[nonenode] ()   [] {};
    \draw (4,-2.5) node[nonenode] ()   [] {};
\end{tikzpicture}
\end{subfigure}
	\caption{A successor $\varphi'$ of a failure pattern $\varphi$ with respect to node $u$. 
		Red nodes are faulty in $\varphi$ and white nodes are correct in it.}
	\label{fig:successor}
\end{figure}

Note that the correct node $w$ in Definition~\ref{def:successor} is well defined as the number of failures satisfies $t<\kappa(G)\leq \delta(G)\leq \deg(u)$, where $\delta(G)$ is the minimum degree of the nodes in $G$. Intuitively, $\sux(\varphi,u)$ is identical to $\varphi$, except that $u$ fails at round $f_u+1$, or it still fails at round~$f_u$ but sends its message to one more correct neighbor before crashing.

Note also that a failure pattern may have different successors, which depends on the choice of the node $u$ that crashes last, and on the choice of the correct neighbor~$w$ of~$u$ in the first and third cases of Definition~\ref{def:successor}. 
A correct neighbor~$w$ of~$u$ in Definition~\ref{def:successor} is called a \emph{witness} of the pair  $(\varphi,\varphi')$. 

Still using the notations of Definition~\ref{def:successor}, let us set $f''_u=f'_u$ in case~1, and $f''_u=f_u$ in cases~2 and~3. 
At the end of round~$f''_u$, there is at most one correct node with different views in $\varphi$ and $\sux(\varphi,u)$. 
The only correct node may have different views in $\varphi$ and $\varphi'=\sux(\varphi,u)$ at the end of round $f''_u$ is the \emph{witness} of the pair  $(\varphi,\varphi')$. 
Before applying the notion of successor to derive our lower bound, let us observe the following. 

\begin{lemma} \label{lem:patternPhi_v}
	For every node $v$, there exists a failure pattern $\varphi \in \Phi^\star_v$ such that no node $u\neq v$ fails at round~$1$ in~$\varphi$, and $\ecc(v, \varphi) \geq \radius(G,t)$. 
\end{lemma}

\begin{proof}
	By definition of the radius, for every $v \in V$, there exists $\psi \in \Phi^\star_v$ such that $\ecc(v,\psi) \geq \radius(G,t)$. The failure pattern $\varphi$ is identical to $\psi$, except that, for every node $u \neq v$ that crashes at round~1 in $\psi$, $u$~crashes cleanly at round~$2$ in $\varphi$. 
	We have $\ecc(v,\varphi) = \ecc(v,\psi)$ because every node that crashes later in $\varphi$ than in $\psi$ does not send any message to their neighbors after round~1 which may contain information received from $v$. Thus $\ecc(v, \varphi) \geq \radius(G,t)$. 
\end{proof}

The premises of the following lemma are justified by Lemma~\ref{lem:patternPhi_v}. 

\begin{lemma} \label{lem:connect}
	Let $\varphi\in\FPA$ such that (1)~at most one node crashes at round~$1$, and (2)~if there exists a node $v$ that crashes at round~$1$ in~$\varphi$, then $\varphi \in \Phi^\star_v$ (i.e., $v$ broadcasts despite the fact that it crashes at round~1). For every successor $\varphi'$ of $\varphi$, the following holds: 
	\begin{itemize}
		\item at most one node crashes at round~$1$ in $\varphi'$; 
		\item if there is a node $v$ that crashes at round~$1$ in~$\varphi'$, then $v$ crashes at round~$1$ in~$\varphi$ as well; 
		\item there exists a correct node with the same view in $\varphi$ and $\varphi'$ at the end of round ${\radius(G,t)-1}$.
	\end{itemize}
\end{lemma}

\begin{proof}
	Let $\varphi'$ be a successor of $\varphi$, such that the entry $(u,F_u,f_u)$ of $\varphi$ is replaced by the entry $(u,F'_u,f'_u)$ in $\varphi'$. 
	Let $w$ be a witness for the pair $(\varphi, \varphi')$ with respect to $u$. 
	Using the notations from Definition~\ref{def:successor}, let $f''_u=f'_u$ in Case~1, and $f''_u=f_u$ in Cases~2 and~3. 
	
	After $f''_u$ rounds, the only correct node that may have different views in $\varphi$ and $\varphi'$ is~$w$. 
	Since $u$ is a node crashing last in $\varphi$, we get that, after round $f''_u$, $w$~needs the same number of rounds in $\varphi$ and $\varphi'$ for broadcasting to all correct nodes. 
	Indeed, all nodes that have not crashed in $\varphi$ nor in $\varphi'$ up to round~$f''_u$ included satisfy: (1)~they are correct nodes in both $\varphi$ and~$\varphi'$, (2)~they  have the same view in both $\varphi$ and~$\varphi'$,
	and (3)~the subgraph of $G$ induced by the correct nodes in $\varphi$ is identical to the subgraph of $G$ induced by the correct nodes in~$\varphi'$. 
	
	Let $R=\radius(G,t)$. We consider two cases, depending on whether $w$ broadcasts or not. 
	
	Let us first consider the case where, assuming that $w$ starts broadcasting at round~$f''_u+1$, $w$~cannot broadcast to all correct nodes during rounds $f''_u+1,\dots,R-1$ under the failure patterns~$\varphi'$ and $\varphi$. 
	That is, under $\varphi'$, some node~$s$ does not receive $\view(w, f''_u,\varphi')$ during rounds $f''_u+1,\ldots,R-1$. 
	As a consequence, this node $s$ does not detect any difference between $\view(w, f''_u,\varphi)$ and $\view(w, f''_u,\varphi')$. 
	It follows that $s$ has the same view in $\varphi$ and $\varphi'$ at the end of $R-1$ rounds.
	
	Consider now the case where, assuming that $w$ starts broadcasting at round~$f''_u+1$, $w$~does succeed to broadcast to all correct nodes during rounds $f''_u+1,\ldots,R-1$ under the failure patterns $\varphi'$ and $\varphi$. 
	Since no node fails after round $f''_u$ in both $\varphi$ and~$\varphi'$, a causal path from $w$ to a node $s$ in rounds $f''_u+1,\ldots,R-1$ is also a causal path from $s$ to $w$ in rounds $f''_u+1,\ldots,R-1$.
	At the end of round $R-1$, every correct node can thus send to $w$ its view at the end of round~$f''_u$. 
	Since no node $s \neq v$ fails at round~$1$, every node $s\neq v$~does send its input to some correct neighbor during round~$1$. 
	Therefore, $s \in \view(w,R-1,\varphi)$ and $s \in view(w,R-1,\varphi')$. 
	Since $\varphi \in \Phi^\star_v$, we get that, at the end of round $f''_u$, there exists a correct node~$x$ that heard from~$v$, i.e., such that $v \in \view(x,f''_u, \varphi)$.  
	At the end of round $R-1$, this node $x$ will send $\view(x,f''_u, \varphi)$ to~$w$, so $v \in \view(w,R-1,\varphi)$. Similarly, $v \in \view(w,R-1,\varphi')$.
	As a consequence, $\view(w,R-1,\varphi)=\view(w,R-1,\varphi')$, and $w$ has a same view in both failure patterns after $R-1$ rounds, as claimed. 
	
	Furthermore, at most one node $v$ crashes at round $1$ in $\varphi'$, and $\varphi' \in \Phi^\star_v$, as desired. 
\end{proof}

Using the characterization of Theorem~\ref{thm:new-characterizationCFRRT} of consensus solvability based on the information-flow graph, it is sufficient to prove the following result for establishing our lower bound. 

\begin{lemma} \label{lem:not_dominate}
	The information-flow graph $\IF(G,\radius(G,t)-1)$ has a connected component that is not dominated by any node of~$V$.
\end{lemma}

\begin{proof}
	Let $R=\radius(G,t)$.  For every node $v\in V$, we denote by $\varphi_v$ a failure pattern in $\Phi^\star_v$ such that  $\varphi_v$ contains no node $u\neq v$ that fails at round~$1$, and $\ecc(v, \varphi_v) \geq R$. 
	The existence of $\varphi_v$ is guaranteed by Lemma~\ref{lem:patternPhi_v}. Borrowing the notation from~\cite{CastanedaFPRRT23}, for every failure pattern~$\varphi$, and every $r\geq 1$, let 
	\[
	\config(\varphi,r)=\{(v,\view(v,\varphi,r))\in V(\IF(G,r)) \mid v\in V \; \mbox{is active in $\varphi$ at round $r$}\},
	\]
	where by $v$ is active in $\varphi$ at round $r$, we mean that $v$ has not crashed in $\varphi$ during rounds $1,\dots,r$. 
	It was proved in~\cite{CastanedaFPRRT23} (see Lemma~4 in there) that, for every failure pattern $\varphi$, and every $r\geq 1$, the subgraph of $\IF(G,r)$ induced by the vertices of
	$\config(\varphi,r)$ is connected.   
	
	We now show that, for every $v\in V$, $\config(\varphi_v,R-1)$ and $\config(\varphi_\varnothing,R-1)$ are contained in the same connected component of $\IF(G,R-1)$. 
	Roughly, we shall construct a sequence of intermediate failure patterns from $\varphi_v$ to $\varphi_\varnothing$ such that, for every two consecutive failure patterns $\psi$ and $\psi'$ in the sequence, there is a correct node with the same view in $\psi$ and~$\psi'$. 
	Note that the existence of this node implies that the subgraph of $\IF(G,R-1)$ induced by $\config(\psi,R-1)$, and the subgraph of $\IF(G,R-1)$ induced by $\config(\psi',R-1)$ are included in the same connected component of $\IF(G,R-1)$. 
	
	Let us order the crashing nodes in $\varphi_v$ in a decreasing order of the rounds at which they crash where ties are broken arbitrarily, and let 
	\[
	u_1,\ldots,u_{t_v}
	\]
	be the resulting sequence. 
	We have $t_v\leq t$ and, for every $i\in\{1,\dots,t_v-1\}$, $f_{u_i} \geq f_{u_{i+1}}$.  
	Let us construct  a sequence 
	\[
	S= \psi_0,\ldots,\psi_{\ell}
	\]
	of failure patterns, where $\psi_0 =  \varphi_v$, and $\psi_{\ell} = \varphi_\varnothing$. This sequence is itself the concatenation of sub-sequences $S_i$ for $i=1,\dots,t_v$ such that 
	\[
	S_1 = \psi_0,\dots, \psi_{\ell_1},
	\]
	and, for every $i\in\{2,\dots,t_v\}$, 
	\[
	S_i=\psi_{\ell_{i-1}+1}, \dots,\psi_{\ell_i}
	\]
	with $0\leq \ell_1\leq \ell_2\leq \dots\leq \ell_{t_v}=\ell$. 
	For every sub-sequence~$S_i$, $i\in\{1,\dots,t_v\}$, and for every $j\in\{\ell_{i-1}+1,\dots,\ell_i-1\}$, we set \[\psi_{j+1}=\sux(\psi_j,u_i).\] 
	Moreover, the first failure pattern $\psi_{\ell_{i-1}+1}$ in the sequence~$S_i$ is obtained from $\varphi_v$ by removing the crashing nodes $u_1,\ldots,u_{i-1}$, i.e., these nodes are correct in $\psi_{\ell_{i-1}+1}$. The last failure pattern $\psi_{\ell_i}$ of the sequence~$S_i$ is when the node $u_i$ that crashes last in $\psi_{\ell_i}$ fails at round $R$. 
	
	\begin{claim}\label{claim:same-view}
		For any two consecutive failure patterns $\psi_j$ and $\psi_{j+1}$ in~$S$, there exists a correct node $w_j$ with the same view in both patterns after $R-1$ rounds, that is,  
		\[\view(w_j,\psi_j,R-1)=\view(w_j,\psi_{j+1},R-1).\]
	\end{claim}
	
	\noindent 
	To see why the claim holds, let us first assume that $\psi_j$ and $\psi_{j+1}$ belong to a same sub-sequence~$S_i$. In this case, the claim directly follows from Lemma \ref{lem:connect}. 
	If $\psi_j$ and $\psi_{j+1}$ do not belong to a same sub-sequence~$S_i$, then $\psi_j$ is the last element of a sub-sequence $S_i$, and $\psi_{j+1}$ is the first element of sub-sequence $S_{i+1}$, then the claim follows from the fact that the sets of nodes crashing in $\psi_j$ and $\psi_{j+1}$ during round $r$ are the same, for every $r\in\{1,\dots,R-1\}$. This completes the proof of Claim~\ref{claim:same-view}.
	
	\medbreak 
	From Claim~\ref{claim:same-view}, for any two consecutive failure patterns $\psi_j$ and $\psi_{j+1}$ in~$S$, $\config(\psi_j)$ and $\config(\psi_{j+1})$ belong to the same  connected component of $\IF(G,R-1)$.
	To wrap up, we have shown that, for every $v \in V$, there exists a connected component of $\IF(G,R-1)$ containing both $\config(\varphi_\varnothing)$ and $\config(\varphi_v)$. Recall that $\varphi_v$ is a failure pattern in $\Phi^\star_v$ satisfying that it contains no node different from~$v$ that fails at round~$1$, and $\ecc (v, \varphi_v) \geq R$. At the end of round $R-1$, no node dominates the component that contains $\config(\varphi_\varnothing)$ because, for every node~$v\in V$, $v$~cannot dominates $\config(\varphi_v,R-1)$. 
\end{proof}

Theorem~\ref{theo:main:consensus} directly follows from Lemma~\ref{lem:not_dominate} by application of Theorem~\ref{thm:new-characterizationCFRRT}. 


\section{A generic set agreement algorithm}\label{subsec:detailclassicalset}
\label{subsec:detailed-results-k-set-agreement}
Let $G=(V,E)$ and $t<\kappa(G)$. 
Let $k\geq 1$. In the $k$-set agreement task, as in consensus, every node $v$ of $G$ starts with an input value $x_v$ from a set $I$ of cardinality at least $k+1$, and every correct node $v$ must output a value $y_v\in \{x_u\mid u\in V\}$. 
However, the agreement condition is relaxed compared to consensus. Specifically, $k$-set agreement requires that the set of  values outputted by the correct nodes is of cardinality at most~$k$ (consensus is thus merely $k$-set agreement for $k=1$ ).

For describing our $k$-set agreement algorithm, we need to generalize the notion of graph eccentricity and radius whenever $k\geq 1$ nodes are ``centers'' instead of just one. 
For every set $S \subseteq V$ of size at most $k$, let the eccentricity of $S$ with respect to a failure pattern $\varphi$, denoted by $\ecc(S,\varphi)$, be the minimum number of rounds such that whenever every node in $S$ broadcasts information by flooding the network, every correct node of $G$ under $\varphi$ receives the information sent by \emph{at least one} of the nodes in~$S$. We also extend $\ecc(S,\varphi)$ to a set $\Phi$ of failure patterns, defining 
$
\ecc(S,\Phi)=\max_{\varphi\in\Phi}\ecc(S,\varphi).
$

Note that, for every $v\in V$, the eccentricity $\ecc(v,\varphi)$ of $v$ under $\varphi$ as defined in Section~\ref{subsec:ecc-conn-radius} satisfies $\ecc(v,\varphi)=\ecc(\{v\},\varphi)$ with this generalized definition of eccentricity. Let 
\[
\Phi_{S}^{\infty} = \{ \varphi \in \FPA \mid  \ecc(S,\varphi)=\infty\},
\]
and let $\Phi_{S}^\star =\FPA \smallsetminus \Phi_S^\infty$. 

\begin{definition}
    The \emph{$k$-center $t$-resilient radius} of $G$ is defined as
\[
\radius(G,t,k)
= \min_{\substack{S\subseteq V\\ |S|\leq k}} \ecc(S,\Phi^\star_S) 
= \min_{\substack{S\subseteq V\\ |S|\leq k}} \; \max_{\varphi\in\Phi^\star_S} \; \ecc(S,\varphi). 
\]
\end{definition}

\noindent We show the following.
\begin{theorem}
\label{theo:upper-bound-k-set-agreement}
For every graph~$G$, every $t<\kappa(G)$, and every $k\geq 1$, $k$-set agreement in $G$ can be solved by an oblivious algorithm running in $\radius(G,t,k)$ rounds under the $t$-resilient model. 
\end{theorem}

\subsection{Broadcasting from a set of sources}

In this section, we describe a generic oblivious algorithm for solving $k$-set agreement in an arbitrary graph $G=(V,E)$ under the $t$-resilient model, for $t<\kappa(G)$. 
We then describe two ``warm up'' algorithms,  both being non-adaptive (oblivious).
In the next section,
we describe an adaptive (oblivious) 
algorithm for solving $k$-set agreement,
proving Theorem~\ref{theo:upper-bound-k-set-agreement}.

\subsubsection{Basic facts on sets of sources}
We start by stating a couple of remarks similar to Proposition~\ref{lem:all-receive-iff-one-receive}, that hold for a set $S\subseteq V$ of source nodes instead of just one node~$s\in V$.

\begin{lemma}\label{lem:setagr-all-receive-or-none}
	For every graph $G$, every $t<\kappa(G)$, every set~$S\subseteq V$, and every failure pattern~$\varphi$ in the $t$-resilient model, $\ecc(S,\varphi)<\infty$ if and only if there exists at least one correct node $v$ that becomes aware of the message broadcast from at least one node $u\in S$.    
\end{lemma}

\begin{proof}
	If $\ecc(S,\varphi)<\infty$, then, by definition, all correct nodes in $\varphi$ receive the message of at least one node $u\in S$. Conversely,  if there exists a correct node $v$ that becomes aware of the message broadcast from a node $u\in S$, then, since the graph induced be the correct nodes in $\varphi$ is connected (as $t<\kappa(G)$), node $v$ can eventually broadcast the information received from $u$ to all correct nodes, and thus $\ecc(S,\varphi)<\infty$.
\end{proof}

For every set $S \subseteq V$, recall that $\Phi_{S}^{\infty} = \{ \varphi \in \FPA \mid  \ecc(S,\varphi)=\infty\}$, and $\Phi_{S}^\star =\FPA \smallsetminus \Phi_S^\infty$. 

\begin{corollary}\label{cor:setagr_all_or_none}
	For every graph $G$, every $t < \kappa(G)$, every set $S \subseteq V$, and every failure pattern $\varphi$, after $\ecc(S,\Phi^\star_S)$ rounds of broadcasting under failure pattern~$\varphi$, either every active node has received the information sent from some node in~$S$, or no active node received any information sent from any node in~$S$.
\end{corollary}

\begin{proof}
	Assume by contradiction that there exists a failure pattern $\varphi \in \Phi^{\star}_S$ such that, after $\ecc(S,\Phi^\star_S)$ rounds of broadcasting from $S$ under $\varphi$, there are some active nodes that heard from $S$, and some active nodes that haven't heard from~$S$. Let us then consider the failure pattern $\varphi'$ identical to $\varphi$, except that all nodes that are still active in $\varphi$ after round $\ecc(S,\Phi^\star_S)$ remains correct in $\varphi'$ for all rounds larger than $\ecc(S,\Phi^\star_S)$. By Lemma~\ref{lem:setagr-all-receive-or-none}, $\varphi' \in \Phi_S^{\star}$. At the end of round $\ecc(S,\Phi^\star_S)$ of broadcasting from $S$ under $\varphi'$, there are some active nodes that heard from $S$, and some active nodes that haven't heard from~$S$, which implies that $\ecc(S,\varphi')>\ecc(S,\Phi^\star_S)$, a contradiction. 
\end{proof}

\subsubsection{Greedy algorithm 1}

The first greedy algorithm consists to construct a sequence of subsets of $V$ with size at most~$k$, iteratively, as follows. For $i\geq 1$, we define
\[
S_i=\argmin \{\ecc(S,\Phi^\star_S)\mid S\subseteq V, \; |S|\leq k, \text{ and } S \cap (\cup_{j=1}^{i-1} S_j) = \varnothing \}
\]
Let $r_1\geq 1$ be the smallest integer such that $|\bigcup_{i=1}^{r_1} S_i| \geq t+1$. The $r_1$ sets $S_1, \ldots, S_{r_1}$, have the following properties. They are pairwise disjoint, of cardinality at most $k$, and, for every $i\in\{1,\dots,r_1-1\}$, 
\[
\ecc(S_i,\Phi^\star_{S_i}) \leq \ecc(S_{i+1},\Phi^\star_{S_{i+1}}).
\]
Let us order the vertices in $\bigcup_{i=1}^{r_1} S_i$ such that, for every two vertices $u$ and $v$ in $\bigcup_{i=1}^{r_1} S_i$, 
\[
u \prec v \iff 
(u\in S_i, v\in S_j, \text{ and } i<j) \text{ or } (\{u,v\}\subseteq  S_i, \text{ and } u < v),  
\]
where $u<v$ means that the identifier of $u$ is smaller than the identifier of~$v$. The first greedy algorithm proceeds as follows. 
\begin{enumerate}
	\item Every node $u\in \bigcup_{i=1}^{r_1} S_i$ broadcasts $(u,x_u)$ during $R_1=\ecc(S_{r_1},\Phi^\star_{S_{r_1}})$ rounds.
	\item Every node $v$ outputs $y_v=x_u$ where $u\in \cup_{i=1}^{r_1} S_i$ is the smallest node according to $\prec$ for which $(u,x_u)\in \view(v,R_1)$. 
\end{enumerate}

\begin{proposition}
	Greedy Algorithm 1 solves $k$-set agreement in $R_1=\ecc(S_{r_1},\Phi^\star_{S_1})$ rounds. 
\end{proposition}

\begin{proof}
	By construction, every correct node runs for $R_1$ rounds, so termination is guaranteed. Validity also holds by construction since every node picks an input value as output value. Note that every correct node receives at least one pair $(u,x_u)$ with $u\in \bigcup_{i=1}^{r_1} S_i$ after  $R_1$ rounds, since $S_{r_1}$ is the ``slowest'' set in $S_1,\dots,S_{r_1}$, and $|\bigcup_{i=1}^{r_1} S_i| \geq t+1$. Regarding agreement, let us assume that the algorithm runs under failure pattern $\varphi \in\FPA$, and let $A$ be the set of nodes that are active at round $R_1$ in~$\varphi$ (i.e., that haven't crashed in $\varphi$ up to this round). Corollary~\ref{cor:setagr_all_or_none} implies that for every $i \in \{1,\ldots,r_1\}$, whenever a node has received the pair $(u,x_u)$ from a node $u\in S_i$ by round $R_1$ for some~$i$, every node have received the message from at least one node in $S_i$. The claim follows.
\end{proof}

\subsubsection{Greedy algorithm 2}

The second greedy algorithm is a slight improvement of the first greedy algorithm. We construct again an ordered sequence of sets iteratively, but in a different manner. For $i\geq 1$, we define
\[
S_i=\argmin \{\ecc(S,\Phi^\star_{S_i})\mid S\subseteq V, \; |S|\leq k, \text{ and } S \smallsetminus (\cup_{j=1}^{i-1} S_j) \neq  \varnothing \}.
\]
That is, instead of asking that the next set does not intersect the previous sets, one just require that the next set contains at least one node that is not in any of the previous sets. 

We define $r_2\geq 1$ as the smallest integer such that $|\bigcup_{i=1}^{r_2} S_i| \geq t+1$. 
The $r_2$ sets $S_1, \ldots, S_{r_2}$ have the following properties: They are all of cardinality at most~$k$, and, for every $i\in\{1,\dots,r_2-1\}$, 
\[
\ecc(S_i,\Phi^\star_{S_i}) \leq \ecc(S_{i+1},\Phi^\star_{S_{i+1}}).
\]
Note that, for every vertex $u\in \bigcup_{j=1}^{r_2} S_i$, there exists a unique index $i$ such that $u \in S_i \smallsetminus \bigcup_{j=1}^{i-1} S_j$. Again, we define an ordering $\prec$ of the nodes in~$S$. Let $u\neq v$ be two nodes in~$\bigcup_{i=1}^{r_2} S_i$. We set 
\[
u \prec v \iff 
(u\in S_i \smallsetminus \cup_{\ell=1}^{i-1} S_\ell, v\in S_j \smallsetminus \cup_{\ell=1}^{j-1} S_\ell  \text{ and }  i<j) \text{ or } (\{u,v\}\subseteq S_i \smallsetminus \cup_{\ell=1}^{i-1} S_\ell  \text{ and }  u < v).  
\]
The algorithm proceeds exactly the same as the previous greedy algorithm (just the setting of the sets $S_i$'s and of the number of rounds differs). 
\begin{enumerate}
	\item Every node $u\in \bigcup_{i=1}^{r_2} S_i$ broadcasts $(u,x_u)$ during $R_2=\ecc(S_{r_2},\Phi^\star_{S_{r_2}})$ rounds.
	\item Every node $v$ outputs $y_v=x_u$ where $u\in \cup_{i=1}^{r_2} S_i$ is the smallest node according to $\prec$ for which $(u,x_u)\in \view(v,R_2)$. 
\end{enumerate}

\begin{proposition}
	Greedy algorithm 2 solves $k$-set agreement in $R_2=\ecc(S_{r_2},\Phi^\star_{S_{r_2}})$ rounds. 
\end{proposition}

\begin{proof}
	As for the first greedy algorithm, termination and validy are satisfied by construction. 
	Regarding agreement, thanks to Corollary~\ref{cor:setagr_all_or_none}, we get that, for every $i\in\{1,\dots,r_2\}$, after $R_2 \geq \ecc(S_i,\Phi^\star_{S_i})$ rounds, either each correct node has received a pair $(u,x_u)$ from at least one node $u\in S_i$, or no correct nodes have received messages from nodes in $S_i$. Therefore, if a node $v$ returns the input $x_u$ of some node $u\in S_i \smallsetminus \bigcup_{j=1}^{i-1} S_j$, then every node $v'$ also returns the input $x_{u'}$ of some node $u'\in S_i \smallsetminus \bigcup_{j=1}^{i-1} S_j$.  
\end{proof}

\subsection{Beyond greedy: an adaptive algorithm}

The greedy algorithms described above are naive in the sense that they ignore the sets of failure patterns. In this section, we describe a faster algorithm, which does take into account the failure patterns. 

\subsubsection{The adaptive algorithm}

Let us recall our adaptive algorithm from Section~\ref{subsec:detailclassicalset}. 
Let $\Phi_0 = \FPA$, and, for $i\geq 1$, let
\[
\begin{cases}
	S_i = \argmin\big \{ \ecc(S,\Phi_S^\star \cap \Phi_{i-1}) \mid S \subseteq V \smallsetminus (\cup_{j=1}^{i-1} S_j) \text{ and } |S|\leq k \big \} \\
	\Phi_{i} = \Phi_{S_i}^{\infty} \cap \Phi_{i-1}
\end{cases}
\]
We stops this construction at set $S_r$, $r\geq 1$, as soon as $|\bigcup_{i=1}^r S_i|\geq t+1$. Let us now consider the same ordering of the nodes as the one for the first greedy algorithm, but for the sets $S_1,\dots,S_r$ constructed as above. Namely, for two different nodes $u$ and $v$ of $\bigcup_{i=1}^r S_i$, we set 
\[
u \prec v \iff 
(u\in S_i, v\in S_j, \text{ and } i<j) \text{ or } (\{u,v\}\subseteq  S_i, \text{ and } u < v),  
\]
The adaptive algorithm then performs as follows: 
\begin{enumerate}
	\item Every node $u\in \bigcup_{i=1}^r S_i$ broadcasts $(u,x_u)$ during $R=\ecc(S_1,\Phi^\star_{S_1})$ rounds.
	\item Every node $v$ outputs $y_v=x_u$ where $u\in \bigcup_{i=1}^r S_i$ is the smallest node according to $\prec$ for which $(u,x_u)\in \view(v,R)$. 
\end{enumerate}
Note that the adaptive algorithm performs in $\ecc(S_1,\Phi^\star_{S_1})=\radius(G,t,k)$ rounds, whereas the greedy algorithms perform in $\ecc(S_{r_1},\Phi^\star_{S_{r_1}})$ rounds, and $\ecc(S_{r_2},\Phi^\star_{S_{r_2}})$ rounds, respectively, with their own setting of the sets $S_i$'s. 

\subsubsection{Correctness of the adaptive algorithm}

We are now ready to prove Theorem~\ref{theo:upper-bound-k-set-agreement}, by proving the correctness of our adaptive algorithm.  

\begin{proof}
	To establish the theorem, we show that the adaptive algorithm solves $k$-set agreement in $\radius(G,t,k)=\ecc(S_1,\Phi^\star_{S_1}$ rounds. 
	Let us consider the execution of the algorithm under an arbitrary failure pattern $\varphi \in \FPA$.
	Let $i$ be the smallest index such that some correct node $v$ has received a message $(u,x_u)$ from a node $u\in S_i$. By Lemma~\ref{lem:setagr-all-receive-or-none}, we thus have $\varphi \in \Phi_{S_i}^\star \cap \Phi_{i-1}$, which implies that 
	\[
	\ecc(S_i,\varphi) \leq \ecc(S_i,\Phi_{S_i}^\star \cap \Phi_{i-1}).
	\]
	Let $R=\radius(G,t,k)$. Assuming that $\ecc(S_i,\Phi_{S_i}^\star \cap \Phi_{i-1}) \leq R$, we would get $\ecc(S_i,\varphi)\leq R$, and, thanks to Corollary~\ref{cor:setagr_all_or_none}, it would follows that every correct node received a pair $(u,x_u)$ from at least one node $u\in S_i$ during the first $R$ rounds, which would establish termination and agreement. 
	To show that $\ecc(S_i,\Phi_{S_i}^\star \cap \Phi_{i-1}) \leq R$ indeed holds, it is sufficient to show that, for every $i\in \{1,\dots,r-1\}$, 
	\begin{equation}\label{eq:it-is-decreasing-again} 
		\ecc(S_{i+1},\Phi_{S_{i+1}}^\star \cap \Phi_{i}) < \ecc(S_i,\Phi_{S_i}^\star \cap \Phi_{i-1}). 
	\end{equation}
	To establish Eq.~\eqref{eq:it-is-decreasing-again}, let us define the following notion. 
	
	\begin{figure}[tb]
		\centering
		\centering
\begin{tikzpicture}[scale=0.9]
    \tikzstyle{nonenode}=[circle,fill=white,minimum size=0pt,inner sep=0pt]
    
    \tikzstyle{whitenode}=[draw,circle,fill=white, minimum size=0pt,inner sep=2pt]
    \tikzstyle{bluenode}=[draw,circle,blue,fill=blue, minimum size=0pt,inner sep=2pt]
    
    \draw[dotted] (0.3,0) ellipse (1.1cm and 2.3cm);
    \draw[dotted] (3.7,0) ellipse (1.1cm and 2.3cm);
    \draw[dotted] (-3,0) ellipse (0.8cm and 2.3cm);
    \draw[dotted] (-6,0) ellipse (0.8cm and 2.3cm);
    
    \draw (0,0) node[font=\small] () {$S_i$};
    \draw (-3,0) node[font=\small] () {$S_{i-1}$};
    \draw (-4.5,0) node[font=\small] () {$\ldots$};
    \draw (-6,0) node[font=\small] () {$S_1$};
    \draw (4,0) node[font=\small] () {$\env(S_i)$};
    \draw (-8,0) node[font=\Large] () {$G$};
    
    \draw (0.5,1) node[whitenode] (a) {};
    \draw (0.8,0) node[whitenode] (b) {};
    \draw (0.5,-1) node[whitenode] (c) {};
    \draw (0.3,-2) node[whitenode] (d) {};
    \draw (0.3,2) node[whitenode] (e) {};

    \draw (3.5,1) node[bluenode] (a') {};
    \draw (3.5,-1) node[bluenode] (c') {};
    \draw (3.7,2) node[bluenode] (e') {};
    \draw (3.7,-2) node[bluenode] (d') {};

    \draw[] (a) -- (a');
    \draw[] (b) -- (a');
    \draw[] (c) -- (c');
    \draw[] (b) -- (c');
    \draw[] (e) -- (e');
    \draw[] (d) -- (d');
\end{tikzpicture}
		\caption{An envelop $\env(S_i)$ of a set $S_i\subseteq V$ in $G=(V,E)$.}
		\label{fig:envelop-setagr}
	\end{figure}
	
	\begin{definition}
		For every $i\in\{1,\ldots,r-1\}$ an \emph{envelop} of $S_i$ is a set $\env(S_i)$ satisfying the following conditions (see Fig.~\ref{fig:envelop-setagr}): 
		\begin{enumerate}
			\item $\env(S_i) \subseteq V \smallsetminus \{S_1,\ldots,S_{i}\}$,
			\item $|\env(S_i)| \leq |S_i|$,
			\item for every $v \in S_i$, $\env(S_i)\cap N(v)\neq \varnothing$, and  
			\item for every $v \in \env(S_i)$, $S_i\cap N(v)\neq \varnothing$.   
		\end{enumerate}  
	\end{definition}
	
	\noindent 
	Note that every node $v\in S_i$ has  a neighbor that is not in $S_1,\ldots,S_i$ because 
	\[
	|\cup_{j=1}^{i} S_j| \leq t < \kappa(G).
	\]
	Choosing one such neighbor for each $v\in S_i$ gives a set $\env(S_i)$, so an envelop $\env(S_i)$ does exist and is well defined.
	We next show that, for every $i>1$, 
	\[
	\ecc(S_i,\Phi_{S_i}^\star \cap \Phi_{i-1}) > \ecc(\env(S_i),\Phi_{\footnotesize \env(S_i)}^\star \cap \Phi_{i}).
	\]
	To see why this inequality holds, let us construct, for every $\varphi \in \Phi_{\footnotesize \env(S_i)}^\star \cap \Phi_{i}$, a failure pattern $\varphi' \in \Phi_{S_i}^\star \cap \Phi_{i-1}$ as follows. 
	In $\varphi'$, every node in $\bigcup_{j=1}^{i-1} S_j$ fails cleanly at the first round, and every node in $S_i$ also fails at the first round but manages to send its message to its neighbors in $\env(S_i)$ (and only to them). 
	The other nodes that fail in $\varphi$ also fails in $\varphi'$  but one round latter. 
	Let us show that $\ecc(S_i, \varphi') >\ecc(\env(S_i),\varphi)$.
	
	For every correct node $u\in V$, there exists a node $v\in S_i$ which is the fastest node in $S_i$ to broadcast to $u$ under $\varphi'$. 
	Let $P$ be a shortest causal path from $v$ to $u$ in $\varphi'$. Note that, by the definition of  $\varphi'$, the first message on this path must go from $v$ to a neighbor $w \in \env(S_i)$.
	The sub-path of $P$ from $w$ to $u$ is also a causal path under~$\varphi$, and is one link shorter than~$P$.
	Thus, for every correct node $u$, 
	the number of rounds required for $\env(S_i)$ to broadcast to $u$ under $\varphi$ is strictly smaller than the number of rounds required for $S_i$ to broadcast to $u$ under~$\varphi'$.
	Therefore, for every $\varphi \in \Phi_{\footnotesize \env(S_i)}^\star \cap \Phi_i$, there exists a failure pattern $\varphi' \in \Phi_{S_{i}}^\star \cap \Phi_{i-1}$ such that $\ecc(S_i, \varphi') > \ecc(\env(S_i),\varphi)$.
	As a consequence,  
	\begin{align*}
		\ecc(S_i,\Phi_{S_i}^\star \cap \Phi_{i-1}) & \geq \ecc(S_i, \varphi') \\
		& > \ecc(\env(S_i),\Phi_{\footnotesize \env(S_i)}^\star \cap \Phi_{i}) \\
		& \geq \ecc(S_{i+1},\Phi_{S_{i+1}}^\star \cap \Phi_{i}), 
	\end{align*}
	which completes the proof of Eq.~\eqref{eq:it-is-decreasing-again}, and the proof of the theorem. 
\end{proof}

\subsection{Round complexities of the set agreement algorithms}

We conclude this section by showing that the three algorithms presented in this section can be ordered and separated with respect to their round-complexities. 

\begin{theorem}
	For every graph $G$, every $t\geq 0$, and every $k\geq 1$: 
	\begin{itemize}
		\item Greedy algorithm~2 is at least as fast as Greedy algorithm~1 in~$G$, and
		\item Adaptive algorithm is at least as fast as Greedy algorithm~2 in~$G$.
	\end{itemize}
	Moreover, 
	\begin{itemize}
		\item there exists $G$, $t\geq 0$, and $k\geq 1$ such that Greedy algorithm~2 is faster than Greedy algorithm~1 in~$G$, 
		\item there exists $G$, $t\geq 0$, and $k\geq 1$ such that Adaptive algorithm is faster than Greedy algorithm~2 in~$G$.
	\end{itemize}
\end{theorem}

\begin{proof}
	The adaptive algorithm performs in $\radius(G,t,k)=\ecc(S_1,\Phi^\star_{S_1})$ rounds, so it is at least as fast as both greedy algorithms. 
	
	Let us show that, for every $G,t$ and $k$, the second greedy algorithm is at least as fast as the first greedy algorithm. Assume that the sequence computed by Greedy~1 is $S_1,\ldots,S_{r_1}$, and the  sequence computed by Greedy~2 is $S'_1,\ldots,S'_{r_2}$. Since 
	\[
	|\bigcup_{j=i}^{r_1} S_j| \geq t+1, \text{and} |\bigcup_{j=i}^{r_2-1} S'_j| < t+1,
	\]
	we have 
	\[
	\bigcup_{j=i}^{r_1} S_j \smallsetminus \bigcup_{j=i}^{r_2-1} S'_j \neq \varnothing.
	\]
	Thus there exists $\ell \in \{1,\ldots,r_1\}$ such that $S_{\ell} \smallsetminus \bigcup_{j=i}^{r_2-1} S'_j \neq \varnothing$. We have
	\[
	\ecc(S'_{r_2},\Phi^\star_{S'_{r_2}}) \leq \ecc(S_{\ell},\Phi^\star_{S_{\ell}}) \leq \ecc(S_{r_1},\Phi^\star_{S_{r_1}})
	\]
	So, our second greedy algorithm is at least as fast as our first greedy algorithm, as claimed. 
	
	\begin{figure}[tb]
		\centering

\begin{tikzpicture}[scale=0.7]
   \tikzstyle{circlenode}=[draw,circle,minimum size=30pt,inner sep=0pt]
    \tikzstyle{whitenode}=[draw,circle,fill=white,minimum size=10pt,inner sep=0pt]
 
\draw (0,0) node[whitenode] (a1)   [] {};
\draw (1,0) node[whitenode] (b1)   [] {};
\draw (2,0) node[whitenode] (c1)   [] {};
\draw (3,0) node[whitenode] (d1)   [] {a};
\draw (4,0) node[whitenode] (e1)   [] {};
\draw (5,0) node[whitenode] (f1)   [] {};
\draw (6,0) node[whitenode] (g1)   [] {};

\draw (7,0) node[whitenode] (a2)   [] {};
\draw (8,0) node[whitenode] (b2)   [] {};
\draw (9,0) node[whitenode] (c2)   [] {};
\draw (10,0.5) node[whitenode] (d21)   [] {b};
\draw (10,-0.5) node[whitenode] (d22)   [] {c};
\draw (11,0) node[whitenode] (e2)   [] {};
\draw (12,0) node[whitenode] (f2)   [] {};
\draw (13,0) node[whitenode] (g2)   [] {};

\draw (8.5,3) node[whitenode] (x)   [] {x};
\draw (8.5,-3) node[whitenode] (y)   [] {y};

\draw (a1) edge node [below] {} (b1);
\draw (b1) edge node [below] {} (c1);
\draw (c1) edge node [below] {} (d1);
\draw (d1) edge node [below] {} (e1);
\draw (e1) edge node [below] {} (f1);
\draw (f1) edge node [below] {} (g1);
\draw (g1) edge node [below] {} (a2);
\draw (a2) edge node [below] {} (b2);
\draw (b2) edge node [below] {} (c2);
\draw (c2) edge node [below] {} (d21);
\draw (c2) edge node [below] {} (d22);
\draw (d21) edge node [below] {} (e2);
\draw (d22) edge node [below] {} (e2);
\draw (e2) edge node [below] {} (f2);
\draw (f2) edge node [below] {} (g2);

\draw (x) edge node [below] {} (a1);
\draw (x) edge node [below] {} (b1);
\draw (x) edge node [below] {} (c1);
\draw (x) edge node [below] {} (d1);
\draw (x) edge node [below] {} (e1);
\draw (x) edge node [below] {} (f1);
\draw (x) edge node [below] {} (g1);
\draw (x) edge node [below] {} (a2);
\draw (x) edge node [below] {} (b2);
\draw (x) edge node [below] {} (c2);
\draw (x) edge node [below] {} (d21);
\draw (x) edge node [below] {} (e2);
\draw (x) edge node [below] {} (f2);
\draw (x) edge node [below] {} (g2);

\draw (y) edge node [below] {} (a1);
\draw (y) edge node [below] {} (b1);
\draw (y) edge node [below] {} (c1);
\draw (y) edge node [below] {} (d1);
\draw (y) edge node [below] {} (e1);
\draw (y) edge node [below] {} (f1);
\draw (y) edge node [below] {} (g1);
\draw (y) edge node [below] {} (a2);
\draw (y) edge node [below] {} (b2);
\draw (y) edge node [below] {} (c2);
\draw (y) edge node [below] {} (d22);
\draw (y) edge node [below] {} (e2);
\draw (y) edge node [below] {} (f2);
\draw (y) edge node [below] {} (g2);

\end{tikzpicture}
    
		\caption{A graph $G$ for which the greedy algorithm~2 is faster than the greedy algorithm~1}
		\label{fig: compare1}
	\end{figure}
	
	\medbreak 
	
	To strictly separate the round complexities of the three algorithms, let us fix $t=k=2$.  
	Figure~\ref{fig: compare1} exhibits a graph $G$ such that the greedy algorithm~2 in $G$ is faster than the greedy algorithm~1 in~$G$. Indeed, the greedy algorithm~1 solves $2$-set agreement in $4$ rounds in~$G$, whereas the greedy algorithm~2 solves $2$-set agreement in $3$~rounds only. To see why, let us compute $\ecc(S,\Phi^\star_S)$ for all subsets $S$ of size at most~$k=2$.
	Let us prove that $\ecc(\{ a, b\},\Phi^\star_{\{ a, b\}}) =3$. We consider four cases: 
	\begin{itemize}
		\item $a, b$ are correct: Then $S = \{ a,b\}$ broadcast in 3 rounds.
		\item $a$ and $b$ crash: Then $S$ broadcast in 3 rounds. If $a$ is not crash cleanly at the first round, then either $a$ sends a message to $x$ or $y$ or an other neighbour, $z$. Note that, all nodes in the graph has distance at most $2$ to $z$ or $x$ or $y$.
		If $b$ is not crash cleanly at the first round, then $b$ sends a message to either $x$ or an other neighbour, $z$. 
		\item $a$ crashes, but $b$ is correct: If $x$ crash then $b$ is correct and in two rounds, $y$ hears from $b$. Otherwise, i.e., if $x$ is correct, then in one round $x$ hear from $b$.
		\item $b$ crashes, but $a$ is correct: If $x$ crash, then in one rounds, $y$ hears from $a$. If $x$ correct, then in one round $x$ hear from $a$.
	\end{itemize}
	Similarly, $\ecc(\{ a, c\},\Phi^\star_{\{ a, c\}}) =3$. For any other set $S$ of size at most~$2$, by reasoning on the case when $x$ and $y$ crash, we have $\ecc(S,\Phi^\star_S) \geq 4$.
	Thus, the greedy algorithm~2 constructs $S_1 = \{a,b\}, S_2 = \{a,c\}$, or $S_1 = \{a,c\}, S_2 = \{a,b\}$, and the algorithm does terminate in $3$~rounds.
	Instead, for the greedy algorithm~1, $S_1$ is either $\{a,b\}$ or $\{a,c\}$ because the sets must be pairwise disjoint, and $\ecc(S_2,\Phi^\star_{S_2}) \geq 4$.
	
	Finally, to separate the adaptive algorithm from the greedy algorithm~2, we consider the graph $G$ displayed on Figure~\ref{fig: compare2}.  The greedy algorithm~2 solves $2$-set agreement in at least $4$ rounds, whereas the adaptive algorithm solves $2$-set agreement in $3$ rounds. To see why, let us compute $\ecc(S,\Phi^\star_S)$ for every sets $S$ of size at most~$2$. We have $\ecc(\{a,b\},\Phi^\star_{\{ a, b\}}) =3$. Instead, for any other set~$S$ of size at most~$2$, $\ecc(S,\Phi^\star_S) \geq 4$.
	Thus, the greedy algorithm~2 will pick $S_1 = \{a,b\}$, and some $S_2$ with $ecc(S_2,\Phi^\star_{S_2}) \geq 4$. As a consequence, the greedy algorithm~2 terminates in at least $4$~rounds.
\end{proof}

\begin{figure}[tb]
	\centering

\begin{tikzpicture}[scale=0.7]
   \tikzstyle{circlenode}=[draw,circle,minimum size=30pt,inner sep=0pt]
    \tikzstyle{whitenode}=[draw,circle,fill=white,minimum size=10pt,inner sep=0pt]
 
\draw (0,0) node[whitenode] (a1)   [] {};
\draw (1,0) node[whitenode] (b1)   [] {};
\draw (2,0) node[whitenode] (c1)   [] {};
\draw (3,0) node[whitenode] (d1)   [] {a};
\draw (4,0) node[whitenode] (e1)   [] {};
\draw (5,0) node[whitenode] (f1)   [] {};
\draw (6,0) node[whitenode] (g1)   [] {};

\draw (7,0) node[whitenode] (a2)   [] {};
\draw (8,0) node[whitenode] (b2)   [] {};
\draw (9,0) node[whitenode] (c2)   [] {};
\draw (10,0) node[whitenode] (d2)   [] {b};
\draw (11,0) node[whitenode] (e2)   [] {};
\draw (12,0) node[whitenode] (f2)   [] {};
\draw (13,0) node[whitenode] (g2)   [] {};

\draw (8.5,3) node[whitenode] (x)   [] {x};
\draw (8.5,-3) node[whitenode] (y)   [] {y};

\draw (a1) edge node [below] {} (b1);
\draw (b1) edge node [below] {} (c1);
\draw (c1) edge node [below] {} (d1);
\draw (d1) edge node [below] {} (e1);
\draw (e1) edge node [below] {} (f1);
\draw (f1) edge node [below] {} (g1);
\draw (g1) edge node [below] {} (a2);
\draw (a2) edge node [below] {} (b2);
\draw (b2) edge node [below] {} (c2);
\draw (c2) edge node [below] {} (d2);
\draw (d2) edge node [below] {} (e2);
\draw (e2) edge node [below] {} (f2);
\draw (f2) edge node [below] {} (g2);

\draw (x) edge node [below] {} (a1);
\draw (x) edge node [below] {} (b1);
\draw (x) edge node [below] {} (c1);
\draw (x) edge node [below] {} (d1);
\draw (x) edge node [below] {} (e1);
\draw (x) edge node [below] {} (f1);
\draw (x) edge node [below] {} (g1);
\draw (x) edge node [below] {} (a2);
\draw (x) edge node [below] {} (b2);
\draw (x) edge node [below] {} (c2);
\draw (x) edge node [below] {} (d2);
\draw (x) edge node [below] {} (e2);
\draw (x) edge node [below] {} (f2);
\draw (x) edge node [below] {} (g2);

\draw (y) edge node [below] {} (a1);
\draw (y) edge node [below] {} (b1);
\draw (y) edge node [below] {} (c1);
\draw (y) edge node [below] {} (d1);
\draw (y) edge node [below] {} (e1);
\draw (y) edge node [below] {} (f1);
\draw (y) edge node [below] {} (g1);
\draw (y) edge node [below] {} (a2);
\draw (y) edge node [below] {} (b2);
\draw (y) edge node [below] {} (c2);
\draw (y) edge node [below] {} (d2);
\draw (y) edge node [below] {} (e2);
\draw (y) edge node [below] {} (f2);
\draw (y) edge node [below] {} (g2);

\end{tikzpicture}
    
	\caption{A graph $G$ for which the adaptive algorithm is faster than the greedy algorithm~2}
	\label{fig: compare2}
\end{figure}

\section{Consensus beyond the connectivity threshold}
\label{subsubsec:consensus-beyond-threhold}

In this section, we extend the consensus algorithm of~\cite{CastanedaFPRRT23} by considering the case where the number~$t$ of failures is unbounded. In particular, $t$~might be larger than the connectivity $\kappa(G)$ of the graph~$G$. The subgraph of $G$ induced by the set of correct nodes may thus be disconnected, split into several connected components. 
As an example, consider the 3-node path $G=(V,E)$ displayed at the top of Fig.~\ref{fig:IFG-disconnected}, and $t$ failures. 
The information flow graph $\IF(G,r,\FPA)$ is connected for every $r\geq 1$, but is not dominated. Our characterization theorem, Theorem~\ref{thm:new-characterizationCFRRT}, applies even for $t\geq \kappa(G)$. 
It follows that consensus in $G$ cannot be solved under $\FPA$ even for $t=1$. The same holds for any graph whenever the failure pattern may disconnect the graph. We therefore consider a weaker variant of consensus, called \emph{local consensus}, adapted to possibly disconnected graphs. 

\begin{figure}[tb]
	\centering
	\includegraphics[width=14cm]{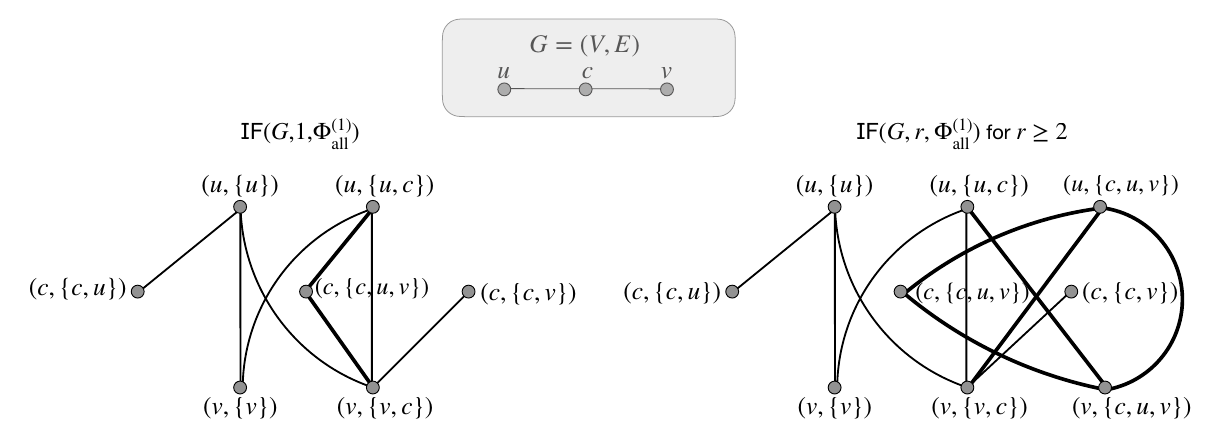}
	\caption{The information flow graphs of a $3$-path, after one round and $r\geq 2$ rounds, where $t=1$ node may fail, potentially disconnecting the graph since $\kappa(G)=1$.}
	\label{fig:IFG-disconnected}
\end{figure}

\subsection{Local consensus}

For every failure pattern  $\varphi$, we define the connected components of $\varphi$ as the connected components of the subgraph of $G$ obtained by removing from $G$ all nodes that crash in~$\varphi$. The set of connected components of $\varphi$  is denoted by $\cc(G,\varphi)$. 

\begin{definition}[Local Consensus]
	Local consensus in a graph $G=(V,E)$ is the problem in which every node $v\in V$ starts with an input value~$x_v$, and every correct node $v\in V$ must decide an output value~$y_v$ such that, (1)~for every failure pattern~$\varphi$, for every connected component $C\in \cc(G,\varphi)$, and for every two correct nodes $u$ and $v$ in~$C$, $y_u=y_v$, and (2)~for every correct node~$v\in V$, there exists~$u\in V$ such that $y_v=x_u$. 
\end{definition}

In other words, local consensus weakens the agreement condition by requiring agreement in each connected components, instead of globally among all correct nodes. However, the validity condition remains the same: every output value of any node $v$ must be equal to the input value of some node~$u$, which may or may not be in the same connected component of the actual failure pattern.  
In particular, if $t\geq\kappa(G)$ but the specific  failure pattern does not disconnect the graph, or if 
$t<\kappa(G)$,
then the definition local consensus coincides with the standard definition of consensus.
Thus, any local consensus algorithm also solves consensus in both cases.

\begin{lemma}
	For every $n$-node graph $G$, and every non-negative integer $t$, local consensus is solvable in $G$ under the $t$-resilient model. 
\end{lemma}

\begin{proof}
	A simple algorithm proceeds in $n-1$ rounds, during which every node $v$ broadcasts the pair $(v,x_v)$. That is, at the first round, every node $v$ sends $(v,x_v)$ to all its neighbors, and, at each subsequent round, every node $v$ forwards to its neighbors all the pairs $(u,x_u)$ received during the previous round. After round $n-1$, every node $v$ outputs $y_v=x_u$ where $u$ is the smallest node (i.e., the node with smallest identifier) received during the execution of the algorithm. The validity condition is satisfied by construction, and we just need to check the agreement condition. For this purpose, let us assume that the execution of the algorithm is subject to failure pattern~$\varphi$. Let $C\in \cc(G,\varphi)$, let $v,v'$ be two distinct nodes in~$C$, and let $(u,x_u)$ be some pair received by~$v$. We claim that $v'$ has also received the same pair $(u,x_u)$. 
	
	To see why, recall that a \emph{causal} path in~$\varphi$ from a node $w$ to a node~$w'$ is a sequence of nodes $a_1,\dots,a_k$ with $a_1=w$, $a_k=w'$, and, for every $i\in\{1,\dots,k-1\}$, $a_{i+1}\in N(a_i)$, $a_i$~has not crashed in $\varphi$ during rounds $1,\dots,i-1$, and if $a_i$ crashes in $\varphi$ at round~$i$, i.e., if $(a_i,F_i,i)\in \varphi$ for some non-empty $F_i\subseteq N(a_i)$, then $a_{i+1}\notin F_i$. The straightforward but crucial observation is that, for every two nodes $w,w'$, if there is a causal path in $\varphi$ from $w$ to $w'$, then this path has length at most $n-1$ (i.e., contains at most $n$ nodes). 
	
	If $v$ has received the pair $(u,x_u)$, then there is a causal path in $\varphi$ from $u$ to $v$. Since $C$ is connected and contains only correct nodes in $\varphi$, it follows that there is also a causal path from $u$ to $v'$. Therefore, $v'$ has also received the pair $(u,x_u)$. In other words, the sets of pairs $(u,x_u)$ received by the two nodes $v$ and $v'$ are identical. Therefore, $y_v=y_{v'}$, and the agreement condition is thus satisfied, which completes the proof. 
\end{proof}

To describe a faster algorithm solving local consensus in any fixed $n$-node graph $G$ under the $t$-resilient model (for any fixed $t\leq n-1$), we need to adapt the notion of eccentricity to failure patterns disconnecting the graph.

\subsection{Eccentricity revisited}

Given a failure pattern $\varphi \in \FPA$, and a connected component $C\in \cc(G,\varphi)$, the eccentricity of $v\in V$ in~$C$ under~$\varphi$, denoted by $\ecc(v,\varphi,C)$, is the number of rounds required to broadcast from $v$ to all nodes in $C$ under~$\varphi$. If some nodes in $C$ cannot receive a message broadcast from $v$ under~$\varphi$, then $\ecc(v,\varphi,C)=\infty$. The following result is a straightforward generalization of Proposition~\ref{lem:all-receive-iff-one-receive} to the setting in which the graph may be disconnected. 

\begin{lemma}\label{lem:manycrash_active}
	For every node~$v$, every failure pattern~$\varphi$, and every connected component $C\in \cc(G,\varphi)$, $\ecc(v,\varphi,C)<\infty$ if and only if there exists at least one node~$w\in C$ that can receive a message broadcast from~$v$ under $\varphi$. In other words, either all nodes of $C$ can receive a message broadcast from~$v$ under~$\varphi$, or none can. 
\end{lemma}

\begin{proof}
	Let $v\in V$, $\varphi\in\FPA$, and $C\in\cc(G,\varphi)$ such that some node $w\in C$ can receive the message broadcast from $v$ under~$\varphi$. Let $w'\in C$ be any node. By definition, there is a path $P$ from $w$ to $w'$ in $C$. Moreover, all nodes in $C$ are correct in~$\varphi$. Therefore, $w'$ will eventually receive the message broadcast from $v$, via $w$, along the path~$P$. 
\end{proof}

\noindent We can then define 
\[
\ecc(v,\varphi) = \max \{\ecc(v,\varphi,C) \;\mid \; C\in\cc(G,\varphi) \text{ and } \ecc(v,\varphi,C)<\infty\}, 
\]
and, for a set $\Phi\subseteq \FPA$ of failure patterns,  
\[
\ecc(v,\Phi) = \max\{\ecc(v,\varphi)\;\mid \; \varphi \in \Phi \text{ and } \ecc(v,\varphi)<\infty\}.  
\]
However, we want to refine the notion of eccentricity to include the connected components instead of just focusing on the failure patterns. For this purpose, let 
\[
\CCA=\{(\varphi,C)\mid \varphi\in \FPA \text{ and } C\in \cc(G,\varphi)\} .
\]
For any $\Omega\subseteq \CCA$, we then define 
\[
\ecc(v,\Omega) = \max \{ \ecc(v,\varphi,C) \; | \; (\varphi,C) \in \Omega \text{ and } \ecc(v,\varphi,C)<\infty \} \}
\]
\color{black}
Finally, the radius of $G$ in the $t$-resilient model is then defined from the eccentricity as for the case $t<\kappa(G)$, that is,  
\[
\radius(G,t)=\min_{v\in V}\ecc(v,\CCA). 
\]

\subsection{The local consensus algorithm}
With the above  definition of radius, adopted to the case ot $t\geq\kappa(G)$, we can finally state main theorem of this section.

\begin{theorem}
	\label{theo:correctness-algo-local-consensus}
	For every connected graph $G=(V,E)$, and every $t\geq 0$, local consensus in $G$ can be solved by an oblivious algorithm running in $\radius(G,t)$ rounds under the $t$-resilient model. 
\end{theorem}

Similarly to the consensus algorithm in~\cite{CastanedaFPRRT23} under the assumption $t<\kappa(G)$, our algorithm for local consensus in the case $t\geq \kappa(G)$ constructs an ordered sequence of $t+1$ nodes as follows. For every node $v\in V$, let 
\[
\Omega^\infty_v=\{(\varphi,C)\in \CCA \mid \ecc(v,\varphi,C)=\infty\}, 
\text{ and } 
\Omega^\star_v=\CCA\smallsetminus \Omega^\infty_v.
\]

\begin{lemma}\label{lem:intersection-of-Omegas-is-empty}
	$\bigcap_{v\in V}\Omega^\infty_v=\varnothing$.
\end{lemma} 

\begin{proof}
	Let us assume for the purpose of contradiction that there exists $(\varphi,C)\in \bigcap_{v\in V}\Omega^\infty_v$. $C$ is a connected component in $\cc(G,\varphi)$, thus $C\neq\varnothing$. Let $u\in C$. Since $C$ is connected and contains only correct nodes in $\varphi$, we have $\ecc(u,\varphi,C)<\infty$, a contradiction. 
\end{proof}

We now have all ingredients to define our algorithm. 
Let us construct a sequence of nodes $s_1,s_2,\dots$ iteratively as follows. Let 
\[
s_1=\argmin_{v\in V}\ecc(v,\Omega^\star_v).   
\]
In other words, we have $\ecc(s_1,\Omega^\star_{s_1})=\ecc(s_1,\CCA)=\radius(G,t)$. Now, for $i\geq 2$, we set
\[
s_{i+1} = \argmin_{v\in V\smallsetminus \{s_1,\dots,s_{i-1}\}} \ecc(v,\Omega_{s_1}^{\infty} \cap \dots \cap \Omega_{s_i}^{\infty} \cap \Omega^\star_{v}),
\]
until one get $s_r$ such that $\Omega_{s_1}^{\infty} \cap \ldots \cap \Omega_{s_r}^{\infty} = \varnothing$. Note that $r$ is well defined as, thanks to Lemma~\ref{lem:intersection-of-Omegas-is-empty}, $\cap_{v\in V}\Omega^\infty_v=\varnothing$.
Our algorithm then performs as follows: 
\begin{enumerate}
	\item Every node $u$ broadcasts $(u,x_u)$ during $\radius(G,t)=\ecc(s_1,\CCA)=\ecc(s_1,\Omega^*_{s_1})$ rounds.
	\item Every node $v$ outputs $y_v=x_{s_i}$ where $s_i$ is the node in the core sequence with smallest index~$i$ for which $(s_i,x_{s_i})\in \view(v,\radius(G,t))$. 
\end{enumerate}

\subsection{Correctness of the local consensus algorithm}

We establish Theorem~\ref{theo:correctness-algo-local-consensus} by proving the correctness of our local consensus algorithm.

\begin{proof}
	The main argument demonstrating the correctness of the core-based algorithm is the fact that, for every $i\in\{1,\dots,r\}$, 
	\begin{equation}\label{eq:yet-it-decreases}
		\ecc(s_i,\Omega_{s_1}^{\infty} \cap \ldots \cap \Omega_{s_{i-1}}^{\infty} \cap \Omega_{s_i}^\star) \leq  \ecc(s_1,\CCA), 
	\end{equation}
	where $\Omega_{s_1}^{\infty} \cap \ldots \cap \Omega_{s_{i-1}}^{\infty}=\varnothing$ for $i=1$. 
	Indeed, let us first assume that Eq.~\eqref{eq:yet-it-decreases} holds. Then, since the sequence $s_1,\ldots,s_r$ that defines our algorithm satisfies 
	\[
	\Omega_{s_1}^{\infty} \cap \ldots \cap \Omega_{s_r}^{\infty} = \varnothing,
	\]
	we have that every correct node hears from at least one node $s_i$, $1\leq i\leq r$, and thus termination is guaranteed. The validity condition holds by construction. For the agreement condition, let us assume that the algorithm performs under failure pattern~$\varphi$, and let  $C\in\cc(G,\varphi)$. There exists $i\in\{1,\dots,r\}$ such that 
	\[
	C \in \Omega_{s_1}^{\infty} \cap \ldots \cap \Omega_{s_{i-1}}^{\infty} \cap \Omega_{s_i}^\star,
	\]
	and thus $s_i$ broadcasts in $C$ under $(G,\varphi)$. By lemma~\ref{lem:manycrash_active}, no node in $C$ hear from any node in $s_1,\ldots,s_{i-1}$.
	Moreover, since 
	\[
	\ecc(s_i,\varphi,C) \leq \ecc(s_i,\Omega_{s_1}^{\infty} \cap \ldots \cap \Omega_{s_{i-1}}^{\infty} \cap \Omega_{s_i}^\star) \leq \radius(G,t),
	\]
	we get that all nodes in $C$ hear from $s_i$. Therefore, all nodes in $C$ output $x_{s_i}$.   
	
	It remains to prove Eq.~\eqref{eq:yet-it-decreases}. The proof goes by induction on~$i=1,\dots,r$. The base case $i=1$ is a tautology. For the induction case, let $1\leq i < r$, and let us assume that, for all $1\leq j \leq i$,
	\[
	\ecc(s_{j},\Omega_{s_1}^{\infty} \cap \ldots \cap \Omega_{s_{j-1}}^{\infty} \cap \Omega_{s_{j}}^\star) \leq \ecc(s_1,\CCA).
	\]
	We aim at proving that $\ecc(s_{i+1},\Omega_{s_1}^{\infty} \cap \ldots \cap \Omega_{s_i}^{\infty} \cap \Omega_{s_{i+1}}^\star) < \ecc(s_1,\CCA$.
	Since $i<r$, we have $V\smallsetminus \{s_1,\ldots,s_i\}\neq \varnothing$. Let $w\in V\smallsetminus \{s_1,\ldots,s_i\}$ at minimum distance to the set $\{s_1,\ldots,s_i\}$ such that  $\Omega_{s_1}^{\infty} \cap \ldots \cap \Omega_{s_i}^{\infty} \cap \Omega_{w}^\star \neq \varnothing$.
	There exists $\varphi\in\FPA$, $C\in\cc(G,\varphi)$, and $w' \in C$ such that 
	\[
	\ecc(w,\Omega_{s_1}^{\infty} \cap \ldots \cap \Omega_{s_i}^{\infty} \cap \Omega_{w}^\star) = \dist(w,w',\varphi)
	\]
	where $\dist(w,w',\varphi)$ denotes the smallest number of rounds required such that $w'$ hear from $w$ in $\varphi$. It is sufficient to prove that $\ecc(w,\Omega_{s_1}^{\infty} \cap \ldots \cap \Omega_{s_i}^{\infty} \cap \Omega_{w}^\star) \leq  \ecc(s_1,\CCA)$. For this purpose, let $P$ be a shortest (causal) path from $w$ to $w'$ in~$\varphi$. There exists a neighbor $z$ of $w$ such that 
	\[
	\Omega_{s_1}^{\infty} \cap \ldots\cap  \Omega_{s_i}^{\infty} \cap \Omega_{z}^\star = \varnothing.
	\]
	For every $u\in V$, let $f_u$ be the rounds at which node $u$ fails in $\varphi$. So, let us consider another failure pattern~$\varphi'$ that is identical to $\varphi$ except that (1)~$z$ sends message to $w$ at the first round in~$\varphi'$, and (2)~for every node $u$ in $P$ that receives messages from $w$ at round $f_u-1$ in $\varphi$, $u$ fails one round later in $\varphi'$ compared to $\varphi$. In $\varphi'$, the information can flow from $z$ to $w$, and then follow path $P$ for reaching $w'$.
	Note that $\cc(G,\varphi')=\cc(G,\varphi)$. Also, observe that, under $\varphi'$, $z$~can broadcast in component $C$, so there exist at least one node in $\{s_1,\ldots,s_i\}$ that can broadcast in~$C$ because $\Omega_{s_1}^{\infty} \cap \ldots\cap  \Omega_{s_i}^{\infty} \cap \Omega_{z}^\star = \varnothing$. Let $s_j$ be such a node, say the one with smallest index~$j$. By this choice, we have  $\ecc(s_j,C)<\infty$. 
	
	\begin{figure}[tb]
		\centering
    \begin{tikzpicture}
        \draw[thick] (0,0) circle [radius=1];
        \draw[thick] (10,0) circle [radius=1.5];

        \draw (0,0) node {$s_1,\ldots,s_i$};
        \draw (3,0) node [draw,circle] (z) {z};
        \draw (4,0) node [draw,circle] (w) {w};
        \draw (10.5,1) node (w') {w'};
        \draw (10,-1) node (C) {$C$};
        \draw (5.5,0) node [draw,circle] (u') [] {u'};
        \draw (7,0) node [draw,circle] (u) [] {u};

        \draw[dotted,thick] (1,0) node {} -- (z);
        \draw[thick] (z) -- (w);
        \draw[dotted,thick] (w) -- (u') -- (u);
        \draw[thick,dotted] (u) .. controls (8,0) .. (w') node[midway, below] {P};

        \draw[thick,dotted] (0,1) .. controls (3,1.5) .. (u) node[midway, above] {P'};
        \draw[thick,dotted] (u) .. controls (8,1) .. (w') node[midway, above] {P'};
         \draw[thick,dotted] (z) .. controls (4,-0.7) .. (u') node[midway, below] {P"};
        \draw[thick,dotted] (u') .. controls (8,-1) .. (w') node[midway, below] {P"};
    \end{tikzpicture}
		\caption{Illustration of the proof of Theorem~\ref{theo:correctness-algo-local-consensus}.}
		\label{fig:proof-theo-local-consensus}
	\end{figure}
	
	Let us now prove that 
	\begin{equation}\label{eq:first-inequality}
		\ecc(s_j,w',\varphi') \geq \dist(z,w',\varphi').
	\end{equation}
	Let $P'$ be a shortest causal path from $s_j$ to $w'$ in $\varphi'$. 
	Since $s_j$ cannot broadcast in $C$ under~$\varphi$, there exists at least one node $u$ belonging to both $P$ and $P'$ that fails  in $\varphi'$ later than in $\varphi$, and $u$ receives messages from $s_j$ at round $f_u$ under $\varphi'$. Since $u$ fails one round later in $\varphi'$ compared to $\varphi$, $u$ receives messages from $w$ at round $f_u-1$ in $\varphi$. So, $u$ receives messages from $z$ at the end of round $f_u$ in $\varphi'$. From node $u$, a message can follow the path $P'$ to reach~$w'$. As a consequence, $\dist(s_j,w',\varphi') \geq \dist(z,w',\varphi')$, as claimed.
	Similarly, let us prove that 
	\begin{equation}\label{eq:second-inequality}
		\dist(z,w',\varphi') > \dist(w,w',\varphi).
	\end{equation}
	Let $P''$ be a shortest causal path from $z$ to $w'$ in $\varphi'$. 
	Again, there must  exist at least one node $u'$ in both $P$ and $P''$ that fails  in $\varphi'$ later than in~$\varphi$, and $u'$ receives messages from $z$ at round $f_{u'}$ under $\varphi'$. Note that $u'$ receives message from $w$ at round $f_{u'}-1$ in $\varphi$.
	Therefore, in $\varphi$, a message from $w$ can follow the path $P"$ from $u'$ for reaching~$w'$. As a consequence, the path $P''$ from $u'$ to $w'$ is at least as long as the path $P$ from $u'$ to $w'$. In addition, the path $P''$ from $z$ to $u'$ is (strictly) longer than the path $P$ from $w$ to $u'$. Thus, $\dist(z,w',\varphi') > \dist(w,w',\varphi)$.
	
	Combining Eq.~\eqref{eq:first-inequality} and~\eqref{eq:second-inequality}, we get  $\dist(s_j,w',\varphi')>\dist(w,w',\varphi)$. 
	By the definition of~$s_j$, we have $C \in \Omega_{s_1}^{\infty} \cap \ldots \cap \Omega_{s_{j-1}}^{\infty} \cap \Omega_{s_{j}}^\star$. As a consequence, 
	\begin{align*}
		\ecc(s_j,\Omega_{s_1}^{\infty} \cap \ldots \cap \Omega_{s_{j-1}}^{\infty} \cap \Omega_{s_{j}}^\star) & \geq \ecc(s_j,\varphi,C) \\
		& \geq \dist(s_j,w',\varphi') \\
		& > \dist(w,w',\varphi) \\
		& = \ecc(w,\Omega_{s_1}^{\infty} \cap \ldots \cap \Omega_{s_i}^{\infty} \cap \Omega_{w}^\star). 
	\end{align*}
	By the definition of $s_{i+1}$,  and thanks to the induction hypothesis, we get that  
	\[
	\ecc(s_1,\CCA) \geq ecc(s_{i+1},\Omega_{s_1}^{\infty} \cap \ldots \cap \Omega_{s_i}^{\infty} \cap \Omega_{s_{i+1}}^\star),
	\]
	which completes the proof of the induction steps, and thus the proof of Theorem~\ref{theo:correctness-algo-local-consensus}. 
\end{proof}

\section{Set agreement beyond the connectivity threshold}

As for consensus, we extend $k$-set agreement to \emph{local} $k$-set agreement in the $t$-resilient model, as follows. 

\begin{definition}[Local Set Agreement]
	For every $k\geq 1$, local $k$-set agreement in a graph $G=(V,E)$ is the problem in which every node $v\in V$ starts with an input value~$x_v$, and every correct node $v\in V$ must decide an output value~$y_v$ such that, (1)~for every failure pattern~$\varphi$ and for every connected component $C\in \cc(G,\varphi)$, $|\{y_v : v \in C\}|\leq k$, and (2)~for every correct node~$v\in V$, there exists~$u\in V$ such that $y_v=x_u$. 
\end{definition}

\subsection{Eccentricity and radius revisited}

Given a failure pattern $\varphi \in \FPA$, and a connected component $C\in \cc(G,\varphi)$, the eccentricity of $S\subseteq V$ in~$C$ under~$\varphi$, denoted by $\ecc(S,\varphi,C)$, is the number of rounds required to broadcast from $S$ to all nodes in $C$ under~$\varphi$, i.e., such that every correct node in $C$ hear from at least one node in $S$. If some nodes in $C$ cannot receive any message broadcast from $S$ under~$\varphi$, then $\ecc(S,\varphi,C)=\infty$.
We can then define 
\[
\ecc(S,\varphi) = \max \{\ecc(S,\varphi,C) \;\mid \; C\in\cc(G,\varphi) \text{ and } \ecc(S,\varphi,C)<\infty\}, 
\]
and, for a set $\Phi \subseteq \FPA$ of failure patterns,  
\[
\ecc(S,\Phi) = \max\{\ecc(S,\varphi)\;\mid \; \varphi \in \Phi \text{ and } \ecc(S,\varphi)<\infty\}.  
\]
Recall that
$
\CCA=\{(\varphi,C)\mid \varphi\in \FPA \text{ and } C\in \cc(G,\varphi)\} .
$
For every $\Omega\subseteq \CCA$, let 
\[
\ecc(S,\Omega) = \max \{ \ecc(S,\varphi,C) \; | \; (\varphi,C) \in \Omega \text{ and } \ecc(S,\varphi,C)<\infty \} \}
\]
As before, for every node $S\subseteq V$, let 
\[
\Omega^\infty_S=\{(\varphi,C)\in \CCA \mid \ecc(S,\varphi,C)=\infty\}, 
\text{ and } 
\Omega^\star_S=\CCA\smallsetminus \Omega^\infty_S.
\]
Finally, we set 
\[
\radius(G,t,k)=\min_{S\subseteq V, |S|\leq k}\ecc(S,\Omega^\star_S).
\]

\subsection{The local set agreement algorithm}

With the revised definitions above, we can introduce the main theorem of this section.

\begin{theorem}
	\label{theo:correctness-algo-local-set-agreement}
	For every connected graph $G=(V,E)$,  every $t\geq 0$, and every $k\geq 1$,  local $k$-set agreement in $G$ can be solved by an oblivious algorithm running in $\radius(G,t,k)$ rounds under the $t$-resilient model. 
\end{theorem}

For our algorithm, we construct a sequence $S_1,s_2,\dots$, where $S_1$ is a \emph{set} of nodes, but, for every $i\geq 2$, $s_i$ is a \emph{single} node. We set   
\[
S_1=\argmin_{S\subseteq V,|S|\leq k}\ecc(S,\Omega^\star_S).   
\]
In other words, we have $\ecc(S_1,\Omega^\star_{S_1})=\ecc(S_1,\CCA)=\radius(G,t,k)$. Now, for $i\geq 1$, we set
\[
s_{i+1} = \argmin_{s\in V\smallsetminus (S_1\cup \{s_2,\dots s_{i}\})} \ecc(s,\Omega_{S_1}^{\infty} \cap \dots \cap \Omega_{s_i}^{\infty} \cap \Omega^\star_{s}),
\]
until we get $s_r$ such that $\Omega_{S_1}^{\infty} \cap \Omega_{s_2}^{\infty} \ldots \cap \Omega_{s_{r}}^{\infty} = \varnothing$. Note that $r$ is well defined as $\cap_{v\in V}\Omega^\infty_v=\varnothing$ (see Lemma~\ref{lem:intersection-of-Omegas-is-empty}).
Let $u\neq v$ be two nodes in~$S_1\cup\{s_2,\dots,s_r\}$. We set 
\[
u \prec v \iff 
(u\in S_1, \; v\notin S_1) 
\text{ or } 
(\{u,v\}\subseteq S_1, \; u < v)
\text{ or } 
(u=s_i, \; v=s_j, \; i<j).
\]
Our algorithm performs as follows: 
\begin{enumerate}
	\item Every node $u\in S_1 \cup \{s_2,\dots,s_r\}$ broadcasts $(u,x_u)$ during $\radius(G,t,k)=\ecc(S_{1},\CCA)$ rounds.
	\item Every node $v$ outputs $y_v=x_u$ where $u\in S_1 \cup \{s_2,\dots,s_r\}$ is the smallest node according to $\prec$ for which $(u,x_u)\in \view(v,\radius(G,t,k))$. 
\end{enumerate}

\subparagraph{Remark.}

It is important to note that our  algorithm for local $k$-set agreement is not a direct extension of our previous algorithm for $k$-set agreement designed for the case  $t<\kappa(G)$. Indeed, in our $k$-set agreement algorithm for $t<\kappa(G)$, the basic notion is an ordered sequence $S_1,\ldots, S_r$ of (disjoint) sets of size at most $k$, where $r$ is the smallest index such that $|\cup_{i=1}^r S_i|\geq t+1$. 
In our local $k$-set agreement algorithm for arbitrary~$t\geq 0$, the central notion is an ordered sequence $S_1,s_2,s_3,\ldots$ where $S_1$ is a set of size at most $k$, but, for every $j\geq 2$, $s_j$ is a single node. 
Nevertheless, our algorithm for local $k$-set agreement also solves standard $k$-set agreement whenever $t<\kappa(G)$. 
Specifically, for $t<\kappa(G)$, the algorithm constructs a sequence $S_1,\ldots,S_r$ with $|S_j|=1$ for all $j\geq 2$. 
The enforcement of using sets that are reduced to single nodes for solving local $k$-set agreement is for pure technical reasons.
Specifically, our proof of correctness for $k$-set agreement does not extend to the case where the graph can be disconnected --- a notion called \emph{envelop} for the sets $S_i$ in the sequence $S_1,\ldots, S_r$ is  well defined only under the condition that the total number of nodes in $\cup_{j=1}^{i}S_j$ is at most $t$, which may not be the case when the failure pattern induces more than one connected components.

\subsection{Correctness of the local set agreement algorithm}

We establish Theorem~\ref{theo:correctness-algo-local-set-agreement} by proving the correctness of our algorithm. 

\begin{proof}
	The validity condition holds by construction. The main argument demonstrating the correctness of the core-based algorithm is the fact that, for every $i\in\{1,\dots,r\}$, 
	\begin{equation}\label{eq:yet-it-decreases-set-agr}
		\ecc(S_i,\Omega_{S_1}^{\infty} \cap \ldots \cap \Omega_{S_{i-1}}^{\infty} \cap \Omega_{S_i}^\star) \leq  \ecc(S_1,\CCA), 
	\end{equation}
	where for $i=2,\ldots,r$, $S_i=\{s_i\}$, and $\Omega_{S_1}^{\infty} \cap \ldots \cap \Omega_{S_{i-1}}^{\infty} = \varnothing$ for $i=1$.
	Indeed, let us first assume that Eq.~\eqref{eq:yet-it-decreases-set-agr} holds. Then, since the sequence $S_1,s_2,\ldots,s_r$ satisfies $\Omega_{S_1}^{\infty} \cap \ldots \cap \Omega_{s_r}^{\infty} = \varnothing$, we have that every correct node hears from at least one node in $\cup_{i=1}^r S_i= S_1\cup(\cup_{i=2}^r \{s_i\})$.  For the agreement condition, let us assume that the algorithm performs under failure pattern~$\varphi$, and let  $C\in\cc(G,\varphi)$. There exists $i\in\{1,\dots,r\}$ such that 
	\[
	C \in \Omega_{S_1}^{\infty} \cap \ldots \cap \Omega_{S_{i-1}}^{\infty} \cap \Omega_{S_i}^\star,
	\]
	and thus $S_i$ broadcasts in $C$ under $(G,\varphi)$. By lemma~\ref{lem:manycrash_active}, no node in $C$ hear from any node in $S_1,\ldots,S_{i-1}$.
	Moreover, since 
	\[
	\ecc(S_i,\varphi,C) \leq \ecc(S_i,\Omega_{S_1}^{\infty} \cap \ldots \cap \Omega_{S_{i-1}}^{\infty} \cap \Omega_{S_i}^\star) \leq \radius(G,t,k),
	\]
	we get that all nodes in $C$ hear from $S_i$. Therefore, every node in $C$ output $x_{s} \in S_i$, and the agreement condition is fulfilled.   
	
	\medbreak 
	
	It remains to prove Eq.~\eqref{eq:yet-it-decreases-set-agr}. The proof goes by induction on~$i=1,\dots,r$. The base case $i=1$ is a tautology. For the induction case, let $1\leq i < r$, and let us assume that, for all $1\leq j \leq i$,
	\[
	\ecc(S_{j},\Omega_{S_1}^{\infty} \cap \ldots \cap \Omega_{S_{j-1}}^{\infty} \cap \Omega_{S_{j}}^\star) \leq \ecc(S_1,\CCA)
	\]
	We aim at proving that $\ecc(S_{i+1},\Omega_{S_1}^{\infty} \cap \ldots \cap \Omega_{S_i}^{\infty} \cap \Omega_{S_{i+1}}^\star) \leq \ecc(S_1,\CCA)$.
	Since $i<r$, we have $V\smallsetminus \{S_1,\ldots,S_i\}\neq \varnothing$. Let $w\in V\smallsetminus \{S_1,\ldots,S_i\}$ at minimum distance to the set $\{S_1,\ldots,S_i\}$ such that  $\Omega_{S_1}^{\infty} \cap \ldots \cap \Omega_{S_i}^{\infty} \cap \Omega_{w}^\star \neq \varnothing$.
	There exists $\varphi\in\FPA$, $C\in\cc(G,\varphi)$, and $w' \in C$ such that 
	\[
	\ecc(w,\Omega_{S_1}^{\infty} \cap \ldots \cap \Omega_{S_i}^{\infty} \cap \Omega_{w}^\star) = \dist(w,w',\varphi)
	\]
	where $\dist(w,w',\varphi)$ denotes the smallest number of rounds required such that $w'$ hear from $w$ in $\varphi$. It is sufficient to prove that $\ecc(w,\Omega_{S_1}^{\infty} \cap \ldots \cap \Omega_{S_i}^{\infty} \cap \Omega_{w}^\star) \leq  \ecc(S_1,\CCA)$. For this purpose, let $P$ be a shortest (causal) path from $w$ to $w'$ in~$\varphi$. There exists a neighbor $z$ of $w$ such that 
	\[
	\Omega_{S_1}^{\infty} \cap \ldots\cap  \Omega_{S_i}^{\infty} \cap \Omega_{z}^\star = \varnothing.
	\]
	For every $u\in V$, let $f_u$ be the rounds at which node $u$ fails in $\varphi$. So, let us consider another failure pattern~$\varphi'$ that is identical to $\varphi$ except that (1)~$z$ sends message to $w$ at the first round in~$\varphi'$, and (2)~for every node $u$ in $P$ that receives messages from $w$ at round $f_u-1$ in $\varphi$, $u$ fails one round later in $\varphi'$ compared to $\varphi$. In $\varphi'$, the information can flow from $z$ to $w$, and then follow path $P$ for reaching $w'$. 
	Note that $\cc(G,\varphi')=\cc(G,\varphi)$. Also, observe that, under $\varphi'$, $z$~can broadcast in component $C$, so there exist at least one node in $S_1,\ldots,S_i$ that can broadcast in~$C$. If there is no node in $S_1$ can broadcast in~$C$ under $\varphi'$, then $s=s_j$ with smallest index~$j$. Otherwise, let $s=s_1$ be the fastest node in $S_1$ broadcast to $C$ in $\varphi'$. By this choice, we have  $\dist(s_j,w',\varphi') \leq \ecc(s_j,C)<\infty$. 
	Similarly to Theorem \ref{theo:correctness-algo-local-consensus}, we now prove that 
	\[
	\ecc(s_j,w',\varphi') \geq \dist(z,w',\varphi')> \dist(w,w',\varphi).
	\]    
	By the definition of $s_j$, we have $C \in \Omega_{S_1}^{\infty} \cap \ldots \cap \Omega_{S_{j-1}}^{\infty} \cap \Omega_{s_{j}}^\star$. Note that if $s_j=s_1\in S_1$, then we have 
	\[
	\ecc(S_1,\Omega_{S_{1}}^\star) \geq \ecc(S_1,C) \geq ecc(s_1,C) \geq \dist(s_1,w',\varphi') > \dist(w,w',\varphi) =  \ecc(w,\Omega_{S_1}^{\infty} \cap \Omega_{w}^\star).
	\] 
	If $s_j \neq s_1$, we have,
	\begin{align*}
		\ecc(s_j,\Omega_{S_1}^{\infty} \cap \ldots \cap \Omega_{S_{j-1}}^{\infty} \cap \Omega_{s_{j}}^\star) & \geq \ecc(s_j,\varphi,C) \\
		& \geq \dist(s_j,w',\varphi') \\
		& > \dist(w,w',\varphi) \\
		& = \ecc(w,\Omega_{S_1}^{\infty} \cap \ldots \cap \Omega_{s_i}^{\infty} \cap \Omega_{w}^\star). 
	\end{align*}
	By the definition of $s_{i+1}$,  and thanks to the induction hypothesis, we get that  
	\[
	\ecc(S_1,\CCA) \geq \ecc(s_{i+1},\Omega_{S_1}^{\infty} \cap \ldots \cap \Omega_{s_i}^{\infty} \cap \Omega_{s_{i+1}}^\star),
	\]
	which completes the proof of the induction steps, and thus the proof of Theorem~\ref{theo:correctness-algo-local-set-agreement}. 
\end{proof}

\section{Conclusion}

In this paper, we have completed the picture for consensus in the $t$-resilient model for arbitrary graphs. That is, we have proved that the consensus algorithm in~\cite{CastanedaFPRRT23} is optimal, i.e., for every graph~$G$ and $t<\kappa(G)$, consensus can be solved by an oblivious algorithm performing in $\radius(G,t)$ rounds under the $t$-resilient model, and no  oblivious algorithm can solve consensus in $G$ in less than $\radius(G,t)$ rounds under the $t$-resilient model. 

We have designed a generic (oblivious) algorithm for $k$-set agreement in arbitrary graph~$G$ performing in $\radius(G,t,k)$ rounds under the $t$-resilient model, for $t<\kappa(G)$.
Moreover, we have extended the study of consensus and $k$-set agreement beyond the connectivity threshold. Specifically, we defined the  \emph{local} consensus and \emph{local} $k$-set agreement tasks, generalizing consensus and $k$-set agreement respectively,
and analyzed generic algorithms for these tasks. 
The technical difficulty of establishing optimality of our algorithms for the local variants of consensus and $k$-set agreement yields from the fact that we miss an analog of our characterization theorem (cf. Theorem~\ref{thm:new-characterizationCFRRT} in Section~\ref{sec:lower-bound-for-consensus}) even for local consensus. 

\subparagraph{Open Problem.} \textit{Is there an oblivious algorithm solving local $k$-set agreement in graph $G$ in less than $\radius(G,t,k)$ rounds under the $t$-resilient model for some graph~$G$, some $k\geq 1$?}

\medskip

Our results open a vast domain for further investigations. 
In particular, what could be said for sets of failure patterns $\Phi$ distinct from~$\FPA$? The case $\Phi_{\mbox{\rm \tiny clean}}$ of clean failures, for which there are no known generic consensus algorithms applying to arbitrary graphs, is particularly intriguing. 
Another intriguing and potentially challenging area for further research is exploring scenarios where  no upper bounds on the number of failing nodes are assumed, by concentrating solely on the set $\Phi_{\mbox{\rm \tiny connect}}$ of failure patterns that do not result in disconnecting the graph.
The main difficulties is that basic results such as Lemma 1 in~\cite{CastanedaFPRRT23} (cf. Proposition~\ref{lem:all-receive-iff-one-receive}) do not hold anymore in this framework. 
Indeed, some ill behaviors that do not occur when the number of failures is bounded from above by the connectivity of the graph, or when the problems are considered in each connected component separately, pop up when the number of failures is arbitrarily large yet preserving connectivity. 

Finally, the design of early-stopping algorithms in the $t$-resilient model for arbitrary graphs is also highly desirable. The early-stopping algorithms in~\cite{ChlebusKOO23} are very promising, but their analysis must be refined to a grain finer than the stretches of the failure patterns, by focusing on, e.g., eccentricities and radii. 

\bibliography{ref}

\begin{thebibliography}{10}

\bibitem{aguilera1999simple}
Marcos~Kawazoe Aguilera and Sam Toueg.
\newblock A simple bivalency proof that t-resilient consensus requires t+ 1
  rounds.
\newblock {\em Information Processing Letters}, 71(3-4):155--158, 1999.

\bibitem{attiya2004distributed}
Hagit Attiya and Jennifer Welch.
\newblock {\em Distributed computing: fundamentals, simulations, and advanced
  topics}, volume~19.
\newblock John Wiley \& Sons, 2004.

\bibitem{castaneda2021topological}
Armando Casta{\~n}eda, Pierre Fraigniaud, Ami Paz, Sergio Rajsbaum, Matthieu
  Roy, and Corentin Travers.
\newblock A topological perspective on distributed network algorithms.
\newblock {\em Theoretical Computer Science}, 849:121--137, 2021.

\bibitem{CastanedaFPRRT23}
Armando Casta{\~{n}}eda, Pierre Fraigniaud, Ami Paz, Sergio Rajsbaum, Matthieu
  Roy, and Corentin Travers.
\newblock Synchronous \emph{t}-resilient consensus in arbitrary graphs.
\newblock {\em Inf. Comput.}, 292:105035, 2023.

\bibitem{CastanedaMRR17}
Armando Casta{\~{n}}eda, Yoram Moses, Michel Raynal, and Matthieu Roy.
\newblock Early decision and stopping in synchronous consensus: {A}
  predicate-based guided tour.
\newblock In {\em 5th International Conference on Networked Systems - (NETYS)},
  volume 10299 of {\em LNCS}, pages 206--221, 2017.

\bibitem{Charron-BostM09}
Bernadette Charron{-}Bost and Stephan Merz.
\newblock Formal verification of a consensus algorithm in the heard-of model.
\newblock {\em Int. J. Softw. Informatics}, 3(2-3):273--303, 2009.

\bibitem{Charron-BostS09}
Bernadette Charron{-}Bost and Andr{\'{e}} Schiper.
\newblock The heard-of model: computing in distributed systems with benign
  faults.
\newblock {\em Distributed Comput.}, 22(1):49--71, 2009.

\bibitem{chaudhuri1991towards}
Soma Chaudhuri.
\newblock Towards a complexity hierarchy of wait-free concurrent objects.
\newblock In {\em Proceedings of the Third IEEE Symposium on Parallel and
  Distributed Processing}, pages 730--737. IEEE, 1991.

\bibitem{chaudhuri1993tight}
Soma Chaudhuri, Maurice Erlihy, Nancy~A. Lynch, and Mark~R. Tuttle.
\newblock Tight bounds for k-set agreement.
\newblock {\em J. ACM}, 47(5):912–943, September 2000.
\newblock \href {https://doi.org/10.1145/355483.355489}
  {\path{doi:10.1145/355483.355489}}.

\bibitem{ChlebusKOO23}
Bogdan~S. Chlebus, Dariusz~R. Kowalski, Jan Olkowski, and Jedrzej Olkowski.
\newblock Disconnected agreement in networks prone to link failures.
\newblock In {\em 25th International Symposium on Stabilization, Safety, and
  Security of Distributed Systems (SSS)}, volume 14310 of {\em LNCS}, pages
  207--222. Springer, 2023.

\bibitem{coulouma2013characterization}
{\'E}tienne Coulouma and Emmanuel Godard.
\newblock A characterization of dynamic networks where consensus is solvable.
\newblock In {\em International Colloquium on Structural Information and
  Communication Complexity}, pages 24--35. Springer, 2013.

\bibitem{Delporte-Gallet18}
Carole Delporte{-}Gallet, Hugues Fauconnier, Sergio Rajsbaum, and Nayuta
  Yanagisawa.
\newblock A characterization of t-resilient colorless task anonymous
  solvability.
\newblock In {\em 25th International Colloquium on Structural Information and
  Communication Complexity (SIROCCO)}, volume 11085 of {\em LNCS}, pages
  178--192. Springer, 2018.

\bibitem{Delporte-GalletFT09}
Carole Delporte{-}Gallet, Hugues Fauconnier, and Andreas Tielmann.
\newblock Fault-tolerant consensus in unknown and anonymous networks.
\newblock In {\em 29th IEEE International Conference on Distributed Computing
  Systems (ICDCS)}, pages 368--375, 2009.

\bibitem{Dolev82}
Danny Dolev.
\newblock The byzantine generals strike again.
\newblock {\em J. Algorithms}, 3(1):14--30, 1982.

\bibitem{dolev1983authenticated}
Danny Dolev and H.~Raymond Strong.
\newblock Authenticated algorithms for byzantine agreement.
\newblock {\em SIAM Journal on Computing}, 12(4):656--666, 1983.

\bibitem{Dolev2000}
Shlomi Dolev.
\newblock {\em Self-Stabilization}.
\newblock {MIT} Press, 2000.

\bibitem{FraigniaudLR22}
Pierre Fraigniaud, Patrick Lambein{-}Monette, and Mika{\"{e}}l Rabie.
\newblock Fault tolerant coloring of the asynchronous cycle.
\newblock In {\em 36th International Symposium on Distributed Computing
  (DISC)}, volume 246 of {\em LIPIcs}, pages 23:1--23:22. Schloss Dagstuhl -
  Leibniz-Zentrum f{\"{u}}r Informatik, 2022.

\bibitem{FraigniaudP20}
Pierre Fraigniaud and Ami Paz.
\newblock The topology of local computing in networks.
\newblock In {\em 47th International Colloquium on Automata, Languages, and
  Programming (ICALP)}, volume 168 of {\em LIPIcs}, pages 128:1--128:18.
  Schloss Dagstuhl - Leibniz-Zentrum f{\"{u}}r Informatik, 2020.

\bibitem{HerlihyKR2013}
Maurice Herlihy, Dmitry~N. Kozlov, and Sergio Rajsbaum.
\newblock {\em Distributed Computing Through Combinatorial Topology}.
\newblock Morgan Kaufmann, 2013.

\bibitem{Suomela2020}
Juho Hirvonen and Jukka Suomela.
\newblock {\em Distributed Algorithms}.
\newblock Aalto University, Finland, 2023.

\bibitem{KhanNV19}
Muhammad~Samir Khan, Syed~Shalan Naqvi, and Nitin~H. Vaidya.
\newblock Exact byzantine consensus on undirected graphs under local broadcast
  model.
\newblock In {\em {PODC}}, pages 327--336. {ACM}, 2019.

\bibitem{LunaV22}
Giuseppe Antonio~Di Luna and Giovanni Viglietta.
\newblock Computing in anonymous dynamic networks is linear.
\newblock In {\em 63rd IEEE Annual Symposium on Foundations of Computer Science
  (FOCS)}, pages 1122--1133, 2022.

\bibitem{LunaV23}
Giuseppe Antonio~Di Luna and Giovanni Viglietta.
\newblock Optimal computation in leaderless and multi-leader disconnected
  anonymous dynamic networks.
\newblock In {\em 37th International Symposium on Distributed Computing
  (DISC)}, volume 281 of {\em LIPIcs}, pages 18:1--18:20. Schloss Dagstuhl -
  Leibniz-Zentrum f{\"{u}}r Informatik, 2023.

\bibitem{Lynch96}
Nancy~A. Lynch.
\newblock {\em Distributed Algorithms}.
\newblock Morgan Kaufmann, 1996.

\bibitem{nowak2019topological}
Thomas Nowak, Ulrich Schmid, and Kyrill Winkler.
\newblock Topological characterization of consensus under general message
  adversaries.
\newblock In {\em Proceedings of the 2019 ACM symposium on principles of
  distributed computing}, pages 218--227, 2019.

\bibitem{Peleg2000}
David Peleg.
\newblock {\em Distributed Computing: A Locality-sensitive Approach}.
\newblock SIAM, 2000.

\bibitem{Raynal02}
Michel Raynal.
\newblock Consensus in synchronous systems: A concise guided tour.
\newblock In {\em 9th Pacific Rim International Symposium on Dependable
  Computing (PRDC)}, pages 221--228. IEEE, 2002.

\bibitem{Raynal2010}
Michel Raynal.
\newblock {\em Fault-tolerant Agreement in Synchronous Message-passing
  Systems}.
\newblock Synthesis Lectures on Distributed Computing Theory. Morgan {\&}
  Claypool Publishers, 2010.

\bibitem{RaynalT06}
Michel Raynal and Corentin Travers.
\newblock Synchronous set agreement: A concise guided tour.
\newblock In {\em 12th {IEEE} Pacific Rim International Symposium on Dependable
  Computing (PRDC)}, pages 267--274, 2006.

\bibitem{winkler2024time}
Kyrill Winkler, Ami Paz, Hugo~Rincon Galeana, Stefan Schmid, and Ulrich Schmid.
\newblock The time complexity of consensus under oblivious message adversaries.
\newblock {\em Algorithmica}, pages 1--32, 2024.

\end{thebibliography}

\end{document}